\documentclass[preprint,10pt,authoryear]{elsarticle}
\usepackage{amssymb}
\usepackage{amsthm}
\usepackage{natbib}
\usepackage{amsmath}
\usepackage{bm}
\usepackage{lineno}
\usepackage{hyperref}
\usepackage{caption}
\usepackage{subcaption}
\usepackage{multirow}
\usepackage{placeins}
\usepackage{chngcntr}
\usepackage{lineno}
\usepackage{float}
\usepackage{diagbox}

\newtheorem{prop}{Proposition}
\usepackage[a4paper, total={6.5in, 9in}]{geometry}
\usepackage{xcolor}

\usepackage{orcidlink}
\usepackage[inline,shortlabels]{enumitem}
\usepackage{ctable}

\DeclareMathAlphabet\mathbfcal{OMS}{cmsy}{b}{n}

\newcommand{\bx}{\boldsymbol{x}}
\newcommand{\bX}{\boldsymbol{X}}
\newcommand{\bN}{\boldsymbol{N}}
\newcommand{\bw}{\boldsymbol{w}}
\newcommand{\bW}{\boldsymbol{W}}
\newcommand{\by}{\boldsymbol{y}}
\newcommand{\bY}{\boldsymbol{Y}}
\newcommand{\bz}{\boldsymbol{z}}
\newcommand{\bv}{\boldsymbol{v}}

\newcommand{\bmu}{\boldsymbol{\mu}}
\newcommand{\bu}{\boldsymbol{u}}
\newcommand{\balpha}{\boldsymbol{\alpha}}
\newcommand{\bSigma}{\boldsymbol{\Sigma}}

\newcommand{\bpi}{\boldsymbol{\pi}}
\newcommand{\bbeta}{\boldsymbol{\beta}}
\newcommand{\bdelta}{\boldsymbol{\delta}}
\newcommand{\boldeta}{\boldsymbol{\eta}}
\newcommand{\btheta}{\boldsymbol{\theta}}

\usepackage{amssymb}

\journal{arXiv}

\begin{document}
\begin{frontmatter}



\author[inst1]{Otto, Arnoldus F. \corref{cor1}\orcidlink{0000-0002-6565-2910}}
\ead{arno.otto@up.ac.za}
\cortext[cor1]{Corresponding author}
\author[inst1]{Bekker, Andri\"ette \orcidlink{0000-0003-4793-5674}}
\author[inst2]{Punzo, Antonio \orcidlink{0000-0001-7742-1821}}
\author[inst3]{Ferreira, Johannes T.\orcidlink{0000-0002-5945-6550}}
\author[inst4]{Tortora, Cristina \orcidlink{0000-0001-8351-3730}}

\affiliation[inst1]{organization={Department of Statistics, University of Pretoria, Pretoria, South Africa}}
\affiliation[inst2]{organization={Department of Economics and Business, University of Catania, Catania, Italy}}
\affiliation[inst3]{organization={School of Statistics and Actuarial Science, University of the Witwatersrand, Johannesburg, South Africa}}
\affiliation[inst4]{organization={Department of Mathematics and Statistics, San José State University, California, United States of America}}

\title{Mixtures of multivariate linear asymmetric Laplace regressions with multiple asymmetric Laplace covariates}


\begin{abstract}
In response to the challenge of accommodating non-Gaussian behaviour in data, the shifted asymmetric Laplace (SAL) cluster-weighted model (SALCWM) is introduced as a model-based method for jointly clustering responses and random covariates that exhibit skewness. Within each cluster, the multivariate SAL distribution is assumed for both the covariates and the responses given the covariates. 
To mitigate the effect of possible atypical observations, a heavy-tailed extension, the contaminated SALCWM (cSALCWM), is also proposed. In addition to the SALCWM parameters, each mixture component has a parameter controlling the proportion of outliers, one controlling the proportion of leverage points, one specifying the degree of outlierness, and another specifying the degree of leverage. The cSALCWM has the added benefit that once the model parameters are estimated and the observations are assigned to components, a more refined intra-group classification in typical points, (mild) outliers, good leverage, and bad leverage points can be directly obtained. An expectation-conditional maximization algorithm is developed for efficient maximum likelihood parameter estimation under this framework. Theoretical identifiability conditions are established, and empirical results from simulation studies and validation via real-world applications demonstrate that the cSALCWM not only preserves the modelling strengths of the SALCWM but also significantly enhances outlier detection and overall inference reliability. The methodology proposed in this paper has been implemented in an \texttt{R} package, which is publicly available at \url{https://github.com/arnootto/ALCWM.}
\end{abstract}


\end{frontmatter}

\section{Introduction} \label{sec:Intro}

Clustering and regression are fundamental tools in statistical modelling, often employed to uncover patterns and relationships within data. 
Traditional finite mixture models have been widely used for clustering, yet they typically assume normality and fail to incorporate dependencies between variables. 
When a clear regression relationship exists among variables, it becomes crucial to account for these dependencies to extract meaningful insights.

Finite Mixture of Regressions Model (MRM; \citet{desarbo1988maximum}) is a well-established approach that models such dependencies while assuming fixed covariates. 
However, a key limitation of MRM is that it does not consider the distribution of covariates themselves, making it less flexible in scenarios where covariates contribute to the group structure. 

To overcome these constraints, the Cluster-Weighted Model (CWM), first introduced by \citet{gershenfeld1997nonlinear}, provides a more comprehensive approach by explicitly modelling the distribution of covariates. 
This model, also known as a finite mixture of regressions with random covariates, offers enhanced flexibility by jointly modelling both response variables and covariates within a clustering framework. 
Over the years, CWMs have been extensively studied and extended in various directions; examples are \citet{ingrassia2015generalized}, \citet{ingrassia2016decision}, \citet{galimberti2020note}, \citet{povcuvca2020modeling}, \citet{punzo2021multivariate},
and \citet{perrone2024parsimonious}.
However, as noted by \citet{gallaugher2022multivariate}, relatively little attention has been devoted to incorporating skewed distributions, limiting their ability to model skewness often present in real-world data.

Herein, we extend this branch of literature by considering the multivariate shifted asymmetric Laplace (SAL) distribution \citep{franczak2013mixtures} to model scenarios where, in each cluster, both the responses given the covariates and random covariates exhibit skewness, respectively. 
The SAL has gained popularity as a viable alternative model for data exhibiting non-Gaussian behaviour, offering a peaked, skewed distribution, with tails heavier than the multivariate normal distribution. 
These properties make it particularly well-suited for applications across multiple disciplines \citep{morris2013dimension, tortora2024laplace, mclaughlin2024unsupervised}.

Real data are, however, often contaminated by atypical values that affect the estimation of the model parameters or, in a regression context, the regression coefficients. 
Improper imposition of CWMs in the presence of these atypical values can underestimate standard errors and overstate the significance of the regression coefficients, which could lead to misleading inference.  

This raises the question: how should these atypical observations be handled? 
To answer this, it is important to note that atypical values are generally divided into two broad categories: 
\begin{enumerate}
    \item Mild atypical values: observations sampled from some population different or even far from the assumed model.
    \item Gross atypical values: observations that cannot be modelled by a distribution as they are unpredictable.
\end{enumerate}
In the presence of gross atypical values, the recommended approach is to eliminate the observations or choose a suitable method for suppressing them \citep{barnett1994outliers}. 
For mild atypical values, which is of specific interest in this paper, it is usually recommended to use a heavy-tailed model that is flexible enough to accommodate the atypical data points \citep{ritter2014robust}. 
However, as emphasized by \citet{davies1993identification}, atypical observations should be defined with respect to a reference distribution. 
That is, the shape of the good (typical or non-atypical) points has to be assumed to define what a bad point is, and the region of bad (atypical) points can be defined, e.g., as a region where the density of the reference distribution is low.

\citet{punzo2017robust} introduced the contaminated Gaussian CWM (CGCWM), which provides some resistance to atypical values but is limited by its assumption that all components are symmetric. 
Although incorporating contamination can help mitigate the influence of atypical values, it may misclassify observations near the tails of asymmetric components as atypical. 
Additionally, the CGCWM (and other symmetric CWMs) may overfit the data by assigning multiple symmetric components to represent a single asymmetric one.


Our proposal is twofold: first, we introduce the SAL cluster-weighted model (SALCWM) to jointly model responses and random covariates that exhibit skewness cluster-wise; second, we propose a heavy-tailed extension, the contaminated SALCWM (cSALCWM), which offers greater robustness to atypical values while preserving the ability to capture skewness in the clusters.

The paper is structured as follows. In Section \ref{Section required background}, some key results are introduced that are used to develop the SALCWM and cSALCWM, which are presented in Section \ref{Section Methodological}. 
There, the proposed CWMs are also compared with mixtures of Laplace regressions. 
Sufficient conditions for identifiability are given in Section \ref{Section Identifiability}, while an expectation-conditional maximization algorithm for parameter estimation is outlined in Section \ref{Sec maximum likelihood estimation}, along with additional operational details. 
A sensitivity analysis demonstrating the effectiveness of the cSALCWM in the presence of mild atypical values is presented in Section \ref{Section sensitivity anslys}. 
Finally, Section \ref{Section Application} provides a real-world application of the SALCWM and cSALCWM, and Section \ref{Section conclusion} offers some concluding remarks.

\section{Required background}\label{Section required background}
In this section, we introduce the key results underlying the development of the SALCWM and the cSALCWM. 
Specifically, Section \ref{Section shifted asymmetric laplace} presents the SAL distribution, while the cSAL distribution is given in Section \ref{Section contaminated shifted asymmetric laplce}.
Section \ref{Section general cluster weighted model} outlines the general CWM framework.
\subsection{Shifted asymmetric Laplace distribution} \label{Section shifted asymmetric laplace}

In this paper, we define $ \bW \sim \text{SAL}_p (\bmu, \bSigma,\balpha) $ to be a $p$-variate random vector following a multivariate SAL distribution with location parameter $ \bmu \in \mathbb{R}^p $, skewness parameter $ \balpha \in \mathbb{R}^p $, and $ p \times p $ non-negative definite matrix $ \bSigma $. 
It follows from \cite{franczak2013mixtures} that the density of $ \bW $ can be expressed as
\begin{align}\label{pdf SAL}
f_{\text{SAL}}(\bw ; \bmu, \bSigma,\balpha) 
&= \frac{2 \exp \{(\bw - \bmu)^\top \bSigma^{-1} \balpha \}}{(2 \pi)^{p/2} |\bSigma|^{1/2}} \left( \frac{(\bw - \bmu)^\top \bSigma^{-1} (\bw - \bmu)}{2 + \balpha^\top \bSigma^{-1} \balpha} \right)^{\nu / 2} K_\nu (u),
\end{align}
where 
$\nu = \frac{2 - p}{2}$,
$u = \sqrt{(2 + \balpha^\top \bSigma^{-1} \balpha) (\bw - \bmu)^\top \bSigma^{-1} (\bw - \bmu)}$, and
$ K_{\nu}(\cdot) $ is the modified Bessel function of the third kind with index $ \nu $.  The random vector $ \bW \sim \text{SAL}_p (\bmu, \bSigma,\balpha) $ can be written using the following stochastic representation
\begin{align}\label{eq AL mixture representation}
\bW = \bmu + V \balpha + \sqrt{V} \bN,
\end{align}
where $ V $ follows an exponential distribution with rate 1, i.e., $ V \sim \text{Exp}(1) $, and $ \bN $ follows a $p$-variate {Gaussian (normal) distribution with mean vector $ \mathbf{0} $ and covariance matrix $ \bSigma $, i.e., $ \bN \sim \mathcal{N}_p (\mathbf{0}, \bSigma) $; consequently, $\bW|V=v \sim \mathcal{N}_p \left( \bmu + v \balpha, v \bSigma \right)$. It follows from \eqref{eq AL mixture representation} that $ \bW $ belongs to the class of multivariate normal variance-mean mixtures and that the expected value and covariance of $ \bW $ are given by, respectively,
\begin{align}
\mathbb{E}(\bW) &= \bmu + \balpha, \\
\text{Cov}(\bW) &= \bSigma + \balpha \balpha^\top.
\end{align}
If $ p = 1 $, then the characteristic function of $ W $ corresponds to a univariate asymmetric Laplace distribution, i.e., $ W \sim \text{SAL}_1 (\mu, \phi, \alpha) $, where $ \nu = 1/2 $ and the Bessel function can be written as $ K_{1/2}(u) = \sqrt{\frac{\pi}{2u}} \exp \{-u\} $. As a result, the density of $ W \sim \text{SAL}_1 (\mu,  \phi, \alpha) $ can be expressed as
\begin{align}
f_{\text{SAL}} (w; \mu, \phi, \alpha) &= \frac{1}{\gamma} \exp \left\{ - \frac{|w - \mu|}{\phi^2} - \gamma \, \alpha \, \text{sign}(w - \mu) \right\},
\end{align}
where $ \gamma = \sqrt{\alpha^2 + 2 \phi^2} $ and all other terms are as previously defined (see \citealp{kotz2001laplace}, for details).

\subsection{Contaminated shifted asymmetric Laplace distribution} \label{Section contaminated shifted asymmetric laplce}
\cite{morris2019asymmetric} employed the approach adopted by \cite{punzo2016parsimonious} to define multivariate contaminated normal mixtures, and introduced the multivariate contaminated SAL (cSAL) distribution with density
\begin{align}\label{pdf csal}
    f_{\text{cSAL}} (\bw; \bmu, \bSigma,\balpha, \delta, \eta) = (1-\delta) \underbrace{f_{\text{SAL}} (\bw;\bmu, \bSigma,\balpha)}_{\text{reference}} + \delta \underbrace{f_{\text{SAL}} (\bw;\bmu, \eta {\bSigma}, \sqrt{\eta} \balpha)}_\text{contaminant}, 
\end{align}
where $\delta \in (0, 1)$ denotes the proportion of bad points and $\eta>1$ denotes the degree of contamination. 
Because of the assumption $\eta>1$, it can be interpreted as the increase in variability due to the bad observations. 
The mode and covariance matrix of the contaminant component $f_{\text{SAL}} (\bw;\bmu, \eta {\bSigma},\sqrt{\eta} \balpha)$ of the cSAL distribution in \eqref{pdf csal} are $\bmu$ and $\eta \left( {\bSigma} + \balpha \balpha^\top \right)$, respectively. 
Hence, the contaminant component has the same mode $\bmu$ and inflated covariance matrix (recall that $\eta > 1$) with respect to the good component; moreover, $\bmu$ is also the mode of the cSAL distribution. 

The SAL distribution can be obtained as a limiting case when $\delta\to1^-$ or $\delta\to0^+$, and $\eta\to1^+$. 
Another advantage is that once the parameters are estimated, say $\hat{\bmu},\hat{{\bSigma}}, \hat{\balpha}, \hat{\delta}$, and $\hat{\eta}$, we can determine whether an observation $\bw$ is good or bad with respect to the reference SAL distribution, using the \emph{a posteriori} probability
\begin{align}
    P(\bw \text{ is good}\mid\hat{\bmu},\hat{{\bSigma}}, \hat{\balpha}, \hat{\delta},\hat{\eta}) = \frac{(1-\hat{\delta})f_{\text{SAL}}(\bw;\hat{\bmu},\hat{{\bSigma}}, \hat{\balpha})}{f_{\text{cSAL}}(\bw;\hat{\bmu},\hat{{\bSigma}}, \hat{\balpha},\hat{\delta},\hat{\eta})},
\end{align}
 and $\bw$ will be considered good if $    P(\bw \text{ is good}\mid\hat{\bmu},\hat{{\bSigma}}, \hat{\balpha}, \hat{\delta},\hat{\eta})\geq0.5$, and bad otherwise.
 
\subsection{The general cluster-weighted model}\label{Section general cluster weighted model}
In many applied problems, the random vector $\bW$ is composed of a $d_{\bY}$ variate response vector $\bY$ and a random vector of covariates $\bX$ of dimension $d_{\bX}$, with $d_{\bX}+d_{\bY}=d_{\bW}$, i.e., $\bW=(\bX,\bY)$. 
In general, a CWM  represents the joint distribution of a response variable $\bY$ and a (random) covariate $\bX$ as
\begin{align}\label{eq general cwm}
p(\bx,\by,\btheta)=\sum^G_{g=1}\pi_gf(\by\mid\bx;\btheta_{\bY|g})f(\bx;\btheta_{\bX|g}),
\end{align}
where $\btheta=\{\pi_g,\btheta_{\bX|g},\btheta_{\bY|g};g=1,\dots,G\}$ represents the set of all parameters, and $\pi_g$ are positive weights, with $\sum^G_{g=1}\pi_g=1$. 
Hereafter, we will add a subscript to all parameters to distinguish those related to $\bX$ from those referring to $\bY$.

Unlike classical mixture models, the CWM factorizes the joint distribution of $p(\bx, \by)$ in the $g$-th mixture component into the product of the conditional distribution of $\bY \mid \bX = \bx$ and the marginal distribution of $\bX$, by assuming a parametric functional relationship for the expectation of $\bY \mid \bx$. 
This approach accounts for the dependency of $\bY$ in $\bX$ within each mixture component by allowing the mean, or more generally, some location parameter to depend on $\bx$ via some linear or nonlinear functional relationship. 
The possibility to specify different models for either $f(\by\mid\bx;\btheta_{\bY|g})$ or $f(\bx;\btheta_{\bX|g})$, as well as different linear/nonlinear dependencies for $\bY\mid\bX=\bx$ in each mixture component makes the CWM a very flexible modelling approach. 
The CWM is particularly useful in this context as part of the family of mixture models with random covariates \citep{hennig2000identifiablity}.  

\section{Methodological proposal} \label{Section Methodological}
In this paper, we focus on \eqref{eq general cwm} by assuming that both $f(\by|\bx;\btheta_{\bY|g})$ and $f(\bx;\btheta_{\bX|g})$ are either the SAL or cSAL in \eqref{pdf SAL} and \eqref{pdf csal}, respectively.
Thus, the SALCWM can be written as
\begin{align}\label{pdf asymmetric Laplace cwm}
p(\bx,\by;\bpi,\bmu,\bSigma,\balpha)=\sum_{g=1}^G\pi_gf_{\text{SAL}}(\by;\bmu_{\bY}(\bx;\bbeta_g),\bSigma_{\bY|g},\balpha_{\bY|g})f_{\text{SAL}}(\bx;\bmu_{\bX|g},\bSigma_{\bX|g},\balpha_{\bX|g}),
\end{align}
where $\pi_g$ are positive weights, $\bmu_{\bY}(\bx;\bbeta_g)=\mathrm{E}(\bY\mid\bx,\bbeta_g)=\bbeta_g^\top\bx^*$ denotes the local conditional mean of $\bY|\bX=\bx$, with $\bbeta_g$ being a vector of regression coefficients of dimension $(1+d_{\bX})\times d_{\bY}$ and $\bx^*=(1,\bx)$ to account for the intercepts.
By considering the cSAL distribution in \eqref{pdf csal} instead, the cSALCWM can be written as  
\begin{align}\label{pdf contaminated asymmetric Laplace cwm}
&p(\bx,\by;\bpi,\bmu,\bSigma,\balpha,\bdelta,\boldeta)\notag\\
&=\sum_{g=1}^G\pi_gf_{\text{cSAL}}(\by;\bmu_{\bY}(\bx;\bbeta_g),\bSigma_{\bY|g},\balpha_{\bY|g},\delta_{\bY|g},\eta_{\bY|g})f_{\text{cSAL}}(\bx;\bmu_{\bX|g},\bSigma_{\bX|g},\balpha_{\bX|g},\delta_{\bX|g},\eta_{\bX|g}).
\end{align}
Just as the cSAL distribution can be seen as a heavy-tailed extension of the SAL distribution, the cSALCWM similarly extends the SALCWM. 
As a result, the SALCWM emerges a limiting case of the cSALCWM if both $\delta_{\bX|g}$ and $\delta_{\bY|g}$ tend to $0^+$ or $1^-$, and $\eta_{\bX|g}$ and $\eta_{\bY|g}$ tend to $1^+$.

In regression analysis, atypical values can be distinguished between two types. 
Atypical observations in $\bY|\bx$ indicate model failure and are referred to as (vertical) outliers. 
Meanwhile, atypical observations with respect to $\bX$ are called leverage points. 
It is useful to distinguish between the two types of leverage points: good and bad.  
A bad leverage point is a regression outlier that has an $\bx$ value that is atypical among the values of $\bX$ as well. 
In contrast, a good leverage point is a point that is unusually large or small among the $\bX$ values but is not a regression outlier, i.e. $\bx$ is atypical, but the corresponding $\by$ fits the model quite well. 
A point like this is considered ``good" because it improves the precision of the regression coefficients. 
The cSALCWM allows us to classify each point $(\bx,\by)$ into one of the four categories presented in \tablename~\ref{table categorization of points in a regression analysis}, and we can do this separately for each cluster.

\begin{table}[!ht]
\caption{Categorization of points in a regression analysis.}
\centering
\begin{tabular}{lrr}
\toprule
   \diagbox{Atypical (on $\bY\mid\bx$)}{Leverage (on $\bX$)} & \multicolumn{1}{c}{No} & \multicolumn{1}{c}{Yes}  \\ \midrule
No &  Typical             & Good leverage                 \\
Yes  &   Outlier        &  Bad leverage           \\ \bottomrule
\end{tabular}
\label{table categorization of points in a regression analysis}
\end{table}
\subsection{Comparison with mixtures of Laplace regressions}\label{Sec comparison with mixtures of laplace regression}
The comparison involves mixtures of regression models with SAL errors (SALMRM) and cSAL errors (cSALMRM), which are respectively defined as
\begin{align}\label{eq mixtures of SAL regression}
    p(\by\mid\bx;\bpi,\bmu,\bSigma,\balpha)=\sum_{g=1}^G\pi_gf_{\text{SAL}}\left(\by;\bmu\left(\bx;\bbeta_g\right),\bSigma_{\bY|g},\balpha_{\bY|g}\right)
\end{align}
and
\begin{align}
\label{eq mixtures of cSAL regression}    p(\by\mid\bx;\bpi,\bmu,\bSigma,\balpha,\delta,\eta)=\sum_{g=1}^G\pi_gf_{\text{cSAL}}\left(\by;\bmu\left(\bx;\bbeta_g\right),\bSigma_{\bY|g},\balpha_{\bY|g},\delta_{\bY|g},\eta_{\bY|g}\right).
\end{align}
Both models in \eqref{eq mixtures of SAL regression} and \eqref{eq mixtures of cSAL regression} belong to the class of mixtures of regression models with fixed covariates. 
Consequently, they suffer from the assignment independence property -- that is, the group assignment probabilities do not depend on the covariates  $\bx$.  
Moreover, these models cannot be used to detect local leverage points.

It is important to note that comparisons between \eqref{eq mixtures of SAL regression} and \eqref{pdf asymmetric Laplace cwm}, or between \eqref{eq mixtures of cSAL regression} and \eqref{pdf contaminated asymmetric Laplace cwm}, are not direct. 
This stems from the fact that the SAL and cSAL regression mixtures model the conditional distribution $p(\by|\bx)$, whereas the CWMs model the joint distribution $p(\bx, \by)$. 
While the fixed-covariate regression models do not provide a model for the marginal distribution $p(\bx)$ — and thus cannot yield $p(\bx, \by)$ — it is still possible to derive the conditional distributions $p(\by|\bx)$ from the SALCWM and cSALCWM frameworks. 
For the cSALCWM, for example, by integrating out $\by$ from model \eqref{pdf contaminated asymmetric Laplace cwm} we obtain
\begin{align}\label{pdf mixture cSAL x only}
    p(\bx;\btheta)=\sum_{g=1}^G\pi_gf_{\text{cSAL}}\left(\bx;\bmu_{\bX|g},\bSigma_{\bX|g},\balpha_{\bX|g},\delta_{\bX|g},\eta_{\bX|g}\right),
\end{align}
which is a mixture of cSAL distributions for the $\bX$ only. The ratio of \eqref{pdf contaminated asymmetric Laplace cwm} over \eqref{pdf mixture cSAL x only} yields
\begin{align}\label{eq csal mixture of regression dynamic weights}
    p(\by|\bx;\btheta)=&\frac{1}{{\sum\limits_{j=1}^G}\pi_jf_{\text{cSAL}}\left(\bx;\bmu_{\bX|j},\bSigma_{\bX|j},\balpha_{\bX|j},\delta_{\bX|j},\eta_{\bX|j}\right)}\notag\\
    &\times\sum_{g=1}^G\pi_gf_{\text{cSAL}}\left(\bx;\bmu_{\bX|g},\bSigma_{\bX|g},\balpha_{\bX|g},\delta_{\bX|g},\eta_{\bX|g}\right)f_{\text{cSAL}}\left(\bx;\bmu\left(\bx;\bbeta_g\right),\bSigma_{\bY|g},\balpha_{\bY|g},\delta_{\bY|g},\eta_{\bY|g}\right),
\end{align}
which is the conditional distribution of $\bY|\bx$ from the cSALCWM and can be seen as a mixture of cSAL regression models with (dynamic) weights depending on $\bx$.
The following proposition shows that the family of mixtures of cSAL regression models can be seen as nested in the family of cSALCWM, as defined by \eqref{eq csal mixture of regression dynamic weights}. 
\begin{prop}\label{proposition mix of regression special case of CWM}
    If, in \eqref{eq csal mixture of regression dynamic weights}, $\bmu_{\bX|1}=\dots=\bmu_{\bX|G}=\bmu_{\bX}$, $\bSigma_{\bX|1}=\dots=\bSigma_{\bX|G}=\bSigma_{\bX}$,    $\balpha_{\bX|1}=\dots=\balpha_{\bX|G}=\balpha_{\bX}$,
$\delta_{\bX|1}=\dots=\delta_{\bX|G}=\delta_{\bX}$, and
$\eta_{\bX|1}=\dots=\eta_{\bX|G}=\eta_{\bX}$, then mixtures of the cSAL regression models as defined in \eqref{eq mixtures of cSAL regression}, can be seen as a particular case of the cSALCWM in \eqref{pdf contaminated asymmetric Laplace cwm}.
\end{prop}
\begin{proof}
A proof of this proposition is given in \ref{Section appendix proof proposition 1}.    
\end{proof}
 Similarly, this can easily be shown for the case of \eqref{eq mixtures of SAL regression} and \eqref{pdf asymmetric Laplace cwm}.

\section{Identifiability}
\label{Section Identifiability}

An important point in dealing with the proposed SALCWM and cSALCWM is establishing their identifiability. 
Without identifiability, the parameters might not be estimated and interpreted, and, more generally, the inference might be meaningless. 
The nontrivial concept of identifiability for mixture models is defined by \citet{titterington1985statistical}, which is used in Proposition \ref{proposition identifiability} to show that the cSALCWM is identifiable (up to the label switching problem).
\begin{prop}\label{proposition identifiability}
Let 
\begin{align*}
p(\bx,\by;\btheta)=\sum_{g=1}^G\pi_gf_{\text{\normalfont cSAL}}(\by;\bmu_{\bY}(\bx;\bbeta_g),\bSigma_{\bY|g},\balpha_{\bY|g},\delta_{\bY|g},\eta_{\bY|g})f_{\text{\normalfont cSAL}}(\bx;\bmu_{\bX|g},\bSigma_{\bX|g},\balpha_{\bX|g},\delta_{\bX|g},\eta_{\bX|g})
\end{align*}
and
\begin{align*}
p(\bx,\by;\tilde\btheta)=\sum_{g=1}^{\tilde G}\pi_jf_{\text{\normalfont cSAL}}(\by;\bmu_{\bY}(\bx;\tilde\bbeta_j),\tilde\bSigma_{\bY|j},\tilde\balpha_{\bY|j},\tilde\delta_{\bY|j},\tilde\eta_{\bY|j})f_{\text{\normalfont cSAL}}(\bx;\tilde\bmu_{\bX|j},\tilde\bSigma_{\bX|j},\tilde\balpha_{\bX|j},\tilde\delta_{\bX|j},\tilde\eta_{\bX|j}).
\end{align*}
be two different parameterizations of the cSALCWM given in \eqref{pdf contaminated asymmetric Laplace cwm}, then the equality $p(\bx,\by,\btheta)\equiv p(\bx,\by,\tilde\btheta)$ implies that $G=\tilde G$ and also implies that for each $g\in\{1,\dots,G\}$ there exists an $j\in\{1,\dots,G\}$ such that $\pi_g=\tilde\pi_j,$ $\bmu_{\bX|g}=\tilde{\bmu}_{\bX|j},$ $ \bSigma_{\bX|g}=\tilde{\bSigma}_{\bX|j}, $ $ \balpha_{\bX|g}=\tilde{\balpha}_{\bX|j},$  $ \delta_{\bX|g}=\tilde\delta_{\bX|j}, $ $ \eta_{\bX|g}=\tilde\eta_{\bX|j}$,  $\bbeta_g=\tilde{\bbeta}_j,$ $\bSigma_{\bY|g}=\tilde{\bSigma}_{\bY|j},$ $\balpha_{\bY|g}=\tilde{\balpha}_{\bY|j},$ $\delta_{\bY|g}=\tilde\delta_{\bY|j},$ $\eta_{\bY|g}=\tilde\eta_{\bY|j}$.
\end{prop}
\begin{proof}
    A proof of this proposition is given in \ref{Section Appendix identifiability}.
\end{proof}
Since the SAL distribution is a special case of the generalized hyperbolic (GH) distribution \citet{gallaugher2022multivariate} and \citep{bagnato2024generalized} proved that the GHCWM is identifiable (up to label switching), it follows that the SALCWM in \eqref{pdf asymmetric Laplace cwm} is also identifiable.

\section{Maximum likelihood estimation}\label{Sec maximum likelihood estimation}
In this section, we present an expectation-conditional maximization (ECM) algorithm of \citet{meng1993maximum} in Section \ref{section ecm algorithm for cSALCWM} for maximum likelihood (ML) estimation of the parameters of the more general cSALCWM in \eqref{pdf contaminated asymmetric Laplace cwm}, and an expectation maximization (EM) algorithm \citep{dempster1977maximum} is presented in Section \ref{Section EM algorithm for SALCWM} of the SALCWM given in \eqref{pdf asymmetric Laplace cwm}. 

\subsection{An ECM algorithm for the cSALCWM}\label{section ecm algorithm for cSALCWM}
Let $(\bx_1,\by_1),\dots,(\bx_n,\by_n)$ be an observed sample from the cSALCWM \eqref{pdf contaminated asymmetric Laplace cwm}. 
For the application of the ECM algorithm, it is convenient to view the observed data as incomplete. 
The sources of incompleteness in this context can be categorized into three distinct types. 
The first source, which is a classical issue in the use of mixture models, arises from the uncertainty regarding the component membership of each observation. 
This uncertainty is represented by an indicator vector $\bz_i = (z_{i1}, \dots, z_{iG})$, where $z_{ig} = 1$ if the observation $(\bx_i, \by_i)$ belongs to component $g$, and $z_{ig} = 0$ otherwise.

The second and third sources of incompleteness are specific to the model under consideration. 
These arise from the fact that, for each observation, it is unknown whether it is an outlier and/or a leverage point with respect to a given component $g$. 
This is denoted by the indicator vector $\bu_i = (u_{i1}, \dots, u_{iG})$, where $u_{ig} = 0$ if $(\bx_i, \by_i)$ is not an outlier in component $g$, and $u_{ig} = 1$ otherwise. 
Similarly, the indicator vector $\bv_i = (v_{i1}, \dots, v_{iG})$ is defined, where $v_{ig} = 0$ if $(\bx_i, \by_i)$ is not a leverage point in component $g$, and $v_{ig} = 1$ otherwise.

The complete-data are thus given by $(\bx_1,\by_1,\bz_1,\bu_1,\bv_1),\dots,(\bx_n,\by_n,\bz_n,\bu_n,\bv_n)$ and, from \eqref{pdf contaminated asymmetric Laplace cwm} (and, inherently, \eqref{pdf csal}), the complete-data likelihood function can be written as
\begin{align}\label{eq likelihood}
    L_c(\bpi,\bmu,\bSigma,\balpha,\bdelta,\boldeta)=&\prod_{i=1}^n\prod_{g=1}^G\left\{\pi_g
\left[\left(1-\delta_{\bX|g}\right)f_{\text{SAL}}\left(\bx_i;\bmu_{\bX|g},\bSigma_{\bX|g},\balpha_{\bX|g}\right)\right]^{1-v_{ig}}\right.\notag\\
   &\left. \times\left[\delta_{\bX|g}f_{\text{SAL}}\left(\bx_i;\bmu_{\bX|g},\eta_{\bX|g}\bSigma_{\bX|g},\sqrt{\eta_{\bX|g}}\balpha_{\bX|g}\right)\right]^{v_{ig}}\right.\notag\\
   &\left. \times \left[\left(1-\delta_{\bY|g}\right)f_{\text{SAL}}\left(\by_i;\bmu_{\bY}(\bx_i;\bbeta_g),\bSigma_{\bY|g},\balpha_{\bY|g}\right)\right]^{1-u_{ig}}\right.\notag\\
  &\left. \times\left[\delta_{\bY|g}f_{\text{SAL}}\left(\by_i;\bmu_{\bY}(\bx_i;\bbeta_g),\eta_{\bY|g}\bSigma_{\bY|g},\sqrt{\eta_{\bY|g}}\balpha_{\bY|g}\right)\right]^{u_{ig}}
    \right\}^{z_{ig}}.
\end{align}
The complete-data log-likelihood then follows from \eqref{eq likelihood} as
\begin{align}\label{eq complete log likelihood}
    l_c(\bpi,\bmu,\bSigma,\balpha,\bdelta,\boldeta)=&l_{1c}(\bpi)+l_{2c}(\bdelta_{\bX})+l_{3c}^{\text{good}}(\bmu_{\bX},\bSigma_{\bX},\balpha_{\bX})+l_{3c}^{\text{bad}}(\bmu_{\bX},\bSigma_{\bX},\balpha_{\bX},\boldeta_{\bX})+l_{4c}(\bdelta_{\bY})\notag\\
    &+l_{5c}^{\text{good}}(\bbeta,\bSigma_{\bY},\balpha_{\bY})+l_{5c}^{\text{bad}}(\bbeta,\bSigma_{\bY},\balpha_{\bY},\boldeta_{\bY}),
\end{align}
where $\bpi=(\pi_1,\dots,\pi_G)$, $\bmu_{\bX}=\left(\bmu_{\bX|1},\dots,\bmu_{\bX|G}\right)$, $\bSigma_{\bX}=\left(\bSigma_{\bX|1},\dots,\bSigma_{\bX|G}\right)$,
$\balpha_{\bX}=\left(\alpha_{\bX|1},\dots,\alpha_{\bX|G}\right)$,
$\bdelta_{\bX}=\left(\delta_{\bX|1},\dots,\delta_{\bX|G}\right)$,
$\boldeta_{\bX}=\left(\eta_{\bX|1},\dots,\eta_{\bX|G}\right)$,
$\bbeta=\left(\bbeta_1,\dots,\bbeta_G\right)$,
$\bSigma_{\bY}=\left(\bSigma_{\bY|1},\dots,\bSigma_{\bY|G}\right)$,
$\balpha_{\bY}=\left(\alpha_{\bY|1},\dots,\alpha_{\bY|G}\right)$,
$\bdelta_{\bY}=\left(\delta_{\bY|1},\dots,\delta_{\bY|G}\right)$,
$\boldeta_{\bY}=\left(\eta_{\bY|1},\dots,\eta_{\bY|G}\right)$,
\allowdisplaybreaks
\begin{align*}
   l_{1c}(\bpi)&=\sum_{i=1}^n\sum_{g=1}^Gz_{ig}\mathrm{ln}\pi_g,\\
    l_{2c}(\bdelta_{\bX})&=\sum_{i=1}^n\sum_{g=1}^Gz_{ig}\left[\left(1-v_{ig}\right)\mathrm{ln}\left(1-\delta_{\bX|g}\right)+v_{ig}\mathrm{ln}\delta_{\bX|g}\right],\\
    l_{3c}^{\text{good}}(\bmu_{\bX},\bSigma_{\bX},\balpha_{\bX})&=\sum_{i=1}^n\sum_{g=1}^Gz_{ig}\left(1-v_{ig}\right)\mathrm{ln}f_{\text{SAL}}\left(\bx_i;\bmu_{\bX|g},\bSigma_{\bX|g},\balpha_{\bX|g}\right),\\
    l_{3c}^{\text{bad}}(\bmu_{\bX},\bSigma_{\bX},\balpha_{\bX},\boldeta_{\bX})&=\sum_{i=1}^n\sum_{g=1}^Gz_{ig}v_{ig}\mathrm{ln}f_{\text{SAL}}\left(\bx_i;\bmu_{\bX|g},\eta_{\bX|g}\bSigma_{\bX|g},\sqrt{\eta_{\bX|g}}\balpha_{\bX|g}\right),\\
    l_{4c}(\bdelta_{\bY})&=\sum_{i=1}^n\sum_{g=1}^Gz_{ig}\left[\left(1-u_{ig}\right)\mathrm{ln}\left(1-\delta_{\bY|g}\right)+u_{ig}\mathrm{ln}\delta_{\bY|g}\right],\\
    l_{5c}^{\text{good}}(\bbeta,\bSigma_{\bY},\balpha_{\bY})&=\sum_{i=1}^n\sum_{g=1}^Gz_{ig}\left(1-u_{ig}\right)\mathrm{ln}f_{\text{SAL}}\left(\by_i;\bmu_{\bY}(\bx_i;\bbeta_g),\bSigma_{\bY|g},\balpha_{\bY|g}\right),\\
    l_{5c}^{\text{bad}}(\bbeta,\bSigma_{\bY},\balpha_{\bY},\boldeta_{\bY})&=\sum_{i=1}^n\sum_{g=1}^Gz_{ig}u_{ig}\mathrm{ln}f_{\text{SAL}}\left(\by_i;\bmu_{\bY}(\bx_i;\bbeta_g),\eta_{\bY|g}\bSigma_{\bY|g},\sqrt{\eta_{\bY|g}}\balpha_{\bY|g}\right).
\end{align*}
Computationally, it is more efficient to use the relationship between the SAL and the multivariate normal distribution outlined in \eqref{eq AL mixture representation} to rewrite $l_{c3}^{\text{good}}$, $l_{c3}^{\text{bad}}$, $l_{c5}^{\text{good}}$, and $l_{c5}^{\text{bad}}$ 
as 
\begin{align}
    l_{3c}^{\text{good}}(\bmu_{\bX},\bSigma_{\bX},\balpha_{\bX})&=\sum_{i=1}^n\sum_{g=1}^Gz_{ig}\left(1-v_{ig}\right)\mathrm{ln}\left[f_{\text{N}}\left(\bx_i;\bmu_{\bX|g}+w_{ig}\balpha_{\bX|g},w_{ig}\bSigma_{\bX|g}\right)f_{\text{Exp}}(w_{ig};1)\right],\\
    l_{3c}^{\text{bad}}(\bmu_{\bX},\bSigma_{\bX},\balpha_{\bX},\boldeta_{\bX})&=\sum_{i=1}^n\sum_{g=1}^Gz_{ig}v_{ig}\mathrm{ln}\left[f_{\text{N}}\left(\bx_i;\bmu_{\bX|g}+w_{ig}\sqrt{\eta_{\bX|g}}\balpha_{\bX|g},w_{ig}\eta_{\bX|g}\bSigma_{\bX|g}\right)f_{\text{Exp}}(w_{ig};1)\right],\\
    l_{5c}^{\text{good}}(\bbeta,\bSigma_{\bY},\balpha_{\bY})&=\sum_{i=1}^n\sum_{g=1}^Gz_{ig}\left(1-u_{ig}\right)\mathrm{ln}\left[f_{\text{N}}\left(\by_i;\bmu_{\bY}(\bx_i;\bbeta_g)+w_{ig}\balpha_{\bY|g},w_{ig}\bSigma_{\bY|g}\right)f_{\text{Exp}}(w_{ig};1)\right],\\
    l_{5c}^{\text{bad}}(\bbeta,\bSigma_{\bY},\balpha_{\bY},\boldeta_{\bY})&=\sum_{i=1}^n\sum_{g=1}^Gz_{ig}u_{ig}\mathrm{ln}\left[f_{\text{N}}\left(\by_i;\bmu_{\bY}(\bx_i;\bbeta_g)+w_{ig}\sqrt{\eta_{\bY|g}}\balpha_{\bY|g},w_{ig}\eta_{\bY|g}\bSigma_{\bY|g}\right)f_{\text{Exp}}(w_{ig};1)\right].
\end{align}
where $f_{\text{Exp}}(\cdot,1)$ denotes the PDF of an exponential distribution with rate 1, i.e., $f_{\text{Exp}}(\cdot,1)=e^{-w_{ig}}, w_{ig}>0$.

The ECM algorithm iterates between three steps, an E step and two CM steps, until convergence. 
The only difference from the EM algorithm is that each M-step is replaced by two simpler CM-steps. 
They arise from the partition $\btheta=(\btheta_1,\btheta_2)$, where $\btheta_1= (\bpi,\bmu_{\bX},\bSigma_{\bX},\balpha_{\bX},\bbeta,\bSigma_{\bY},\balpha_{\bY})$ and $\btheta_2=(\boldeta_{\bX},\boldeta_{\bY})$.
\subsubsection*{E-step}
In the E-step, we compute the conditional expectation of the complete-data log-likelihood function, denoted as 
$Q\left(\btheta|\btheta^{(r)}\right)$ 
for the $(r+1)$-th iteration, where 
\begin{align}     Q\left(\btheta|\btheta^{(r)}\right)=&Q_{1}\left(\bpi|\btheta^{(r)}\right)+Q_{2}\left(\bdelta_{\bX}|\btheta^{(r)}\right)+Q_{3}^{\text{good}}\left(\bmu_{\bX},\bSigma_{\bX},\balpha_{\bX}|\btheta^{(r)}\right)+Q_{3}^{\text{bad}}\left(\bmu_{\bX},\bSigma_{\bX},\balpha_{\bX},\boldeta_{\bX}|\btheta^{(r)}\right)\notag\\
&+Q_{4}\left(\bdelta_{\bY}|\btheta^{(r)}\right)+Q_{5}^{\text{good}}\left(\bbeta,\bSigma_{\bY},\balpha_{\bY}|\btheta^{(r)}\right)+Q_{5}^{\text{bad}}\left(\bbeta,\bSigma_{\bY},\balpha_{\bY},\boldeta_{\bY}|\btheta^{(r)}\right)
\end{align} which is in the same order as \eqref{eq complete log likelihood}, where
\begin{align}
Q_1\left(\bpi|\theta^{(r)}\right)=&\sum_{i=1}^n z_{ig}^{(r)}\mathrm{ln}\pi_g,\label{eq Q1}\\
Q_{2}\left(\bdelta_{\bX}|\btheta^{(r)}\right)=& \sum^{n}_{i=1}\sum^{G}_{g=1} z_{ig}^{(r)}\left[\left(1-v_{ig}^{(r)}\right)\mathrm{ln}\left(1-\delta_{\bX|g}\right)+v_{ig}^{(r)}\mathrm{ln}\delta_{\bX|g}\right],\label{eq Q2}\\
Q_3^{\text{good}}\left(\bmu_{\bX},\bSigma_{\bX},\balpha_{\bX}\right)=&-\frac{nd_{\bX}}{2}\mathrm{ln}(2\pi)-\frac{1}{2}\sum^n_{i=1}\sum^G_{g=1}z_{ig}^{(r)}\left(1-v_{ig}^{(r)}\right)\mathrm{ln}\left|\bSigma_{\bX|g}\right|-\frac{d_{\bX}}{2}\sum_{i=1}^n\sum_{g=1}^Gz_{ig}^{(r)}\left(1-v_{ig}^{(r)}\right)\mathrm{E}_{{3_{\bX|ig}}}^{(r)}\notag\\
    &-\frac{1}{2}\sum_{i=1}^n\sum_{g=1}^Gz_{ig}^{(r)}\left(1-v_{ig}^{(r)}\right)\mathrm{E}_{2_{\bX|ig}}^{(r)}\left[(\bx_i-\bmu_{\bX|g})^\top\bSigma^{-1}_{\bX|g}(\bx_i-\bmu_{\bX|g})\right]\notag\\
&+\sum_{i=1}^n\sum_{g=1}^Gz_{ig}^{(r)}\left(1-v_{ig}^{(r)}\right)\left[(\bx_i-\bmu_{\bX|g})^\top\bSigma^{-1}_{\bX|g}\balpha_{\bX|g}\right]\notag\\
&-\frac{1}{2}\sum_{i=1}^n\sum_{g=1}^Gz_{ig}^{(r)}\left(1-v_{ig}^{(r)}\right)\mathrm{E}_{1_{\bX|ig}}^{(r)}\balpha^\top_{\bX|g}\bSigma^{-1}_{\bX|g}\balpha_{\bX|g}-\sum_{i=1}^n\sum_{g=1}^Gz_{ig}^{(r)}\left(1-v_{ig}^{(r)}\right)\mathrm{E}_{1_{\bX|ig}}^{(r)},\\
Q_3^{\text{bad}}\left(\bmu_{\bX},\bSigma_{\bX},\balpha_{\bX},\boldeta_{\bX}\right)=&-\frac{nd_{\bX}}{2}\mathrm{ln}(2\pi)-\frac{1}{2}\sum^n_{i=1}\sum^G_{g=1}z_{ig}^{(r)}v_{ig}^{(r)}\mathrm{ln}\left|\bSigma_{\bX|g}\right| -\frac{d_{\bX}}{2}\sum^n_{i=1}\sum^G_{g=1}z_{ig}^{(r)}v_{ig}^{(r)}\mathrm{ln}\eta_{\bX|g}\notag\\
    &-\frac{d_{\bX}}{2}\sum_{i=1}^n\sum_{g=1}^Gz_{ig}^{(r)}v_{ig}^{(r)}\mathrm{\tilde{E}}_{3_{\bX|ig}}^{{(r)}}
    -\frac{1}{2}\sum_{i=1}^n\sum_{g=1}^Gz_{ig}^{(r)}v_{ig}^{(r)}\mathrm{\tilde{E}}_{2_{\bX|ig}}^{(r)}\frac{1}{\eta_{\bX|g}}(\bx_i-\bmu_{\bX|g})^\top\bSigma^{-1}_{\bX|g}(\bx_i-\bmu_{\bX|g})\notag\\
&+\sum_{i=1}^n\sum_{g=1}^Gz_{ig}^{(r)}v_{ig}^{(r)}\frac{1}{\sqrt{\eta_{\bX|g}}}(\bx_i-\bmu_{\bX|g})^\top\bSigma^{-1}_{\bX|g}\balpha_{\bX|g}\notag\\
&-\frac{1}{2}\sum_{i=1}^n\sum_{g=1}^Gz_{ig}^{(r)}v_{ig}^{(r)}\mathrm{\tilde{E}}_{1_{\bX|ig}}^{(r)}\balpha^\top_{\bX|g}\bSigma^{-1}_{\bX|g}\balpha_{\bX|g}-\sum_{i=1}^n\sum_{g=1}^Gz_{ig}^{(r)}v_{ig}^{(r)}\mathrm{\tilde{E}}_{1_{\bX|ig}}^{(r)},\label{eq Q3 bad}\\
Q_{4}\left(\bdelta_{\bY}|\btheta^{(r)}\right)=& \sum^{n}_{i=1}\sum^{G}_{g=1} z_{ig}^{(r)}\left[\left(1-u_{ig}^{(r)}\right)\mathrm{ln}\left(1-\delta_{\bY|g}\right)+u_{ig}^{(r)}\mathrm{ln}\delta_{\bY|g}\right],\label{eq Q4}\\ Q_5^{\text{good}}\left(\bbeta,\bSigma_{\bY},\balpha_{\bY}\right)=&\frac{nd_{\bY}}{2}\mathrm{ln}(2\pi)-\frac{1}{2}\sum^n_{i=1}\sum^G_{g=1}z_{ig}^{(r)}\left(1-u_{ig}^{(r)}\right)\mathrm{ln}\left|\bSigma_{\bY|g}\right|-\frac{d_{\bY}}{2}\sum_{i=1}^n\sum_{g=1}^Gz_{ig}^{(r)}\left(1-u_{ig}^{(r)}\right)\mathrm{E}_{3_{\bY|ig}}^{{(r)}}\notag\\
    &-\frac{1}{2}\sum_{i=1}^n\sum_{g=1}^Gz_{ig}^{(r)}\left(1-u_{ig}^{(r)}\right)\mathrm{E}_{2_{\bY|ig}}^{(r)}(\by_i-\bmu_{\bY}(\bx_i;\bbeta_g))^\top\bSigma^{-1}_{\bY|g}\left(\by_i-\bmu_{\bY}(\bx_i;\bbeta_g)\right)\notag\\
&+\sum_{i=1}^n\sum_{g=1}^Gz_{ig}^{(r)}\left(1-u_{ig}^{(r)}\right)(\bx_i-\bmu_{\bY}(\bx_i;\bbeta_g))^\top\bSigma^{-1}_{\bY|g}\balpha_{\bY|g}\notag\\
&-\frac{1}{2}\sum_{i=1}^n\sum_{g=1}^Gz_{ig}^{(r)}\left(1-u_{ig}^{(r)}\right)\mathrm{E}_{1_{\bY|ig}}^{(r)}\balpha^\top_{\bY|g}\bSigma^{-1}_{\bY|g}\balpha_{\bY|g}-\sum_{i=1}^n\sum_{g=1}^Gz_{ig}^{(r)}\left(1-v_{ig}^{(r)}\right)\mathrm{E}_{1_{\bY|ig}}^{(r)},\\
Q_5^{\text{bad}}\left(\bbeta,\bSigma_{\bY},\balpha_{\bY},\boldeta_{\bY}\right)=&\frac{nd_{\bY}}{2}\mathrm{ln}(2\pi)-\frac{1}{2}\sum^n_{i=1}\sum^G_{g=1}z_{ig}^{(r)}u_{ig}^{(r)}\mathrm{ln}\left|\bSigma_{\bY|g}\right| -\frac{d_{\bY}}{2}\sum^n_{i=1}\sum^G_{g=1}z_{ig}^{(r)}u_{ig}^{(r)}\mathrm{ln}\eta_{\bY|g}\notag\\
    &-\frac{d_{\bX}}{2}\sum_{i=1}^n\sum_{g=1}^Gz_{ig}^{(r)}v_{ig}^{(r)}\mathrm{\tilde{E}}_{3_{\bY|ig}}^{{(r)}}\notag\\
    &-\frac{1}{2}\sum_{i=1}^n\sum_{g=1}^Gz_{ig}^{(r)}u_{ig}^{(r)}\mathrm{\tilde{E}}_{2_{\bY|ig}}^{(r)}\frac{1}{\eta_{\bY|g}}(\by_i-\bmu_{\bY}(\bx_i;\bbeta_g))^\top\bSigma^{-1}_{\bY|g}(\bx_i-\bmu_{\bY}(\bx_i;\bbeta_g))\notag\\
&+\sum_{i=1}^n\sum_{g=1}^Gz_{ig}^{(r)}u_{ig}^{(r)}\frac{1}{\sqrt{\eta_{\bY|g}}}(\by_i-\bmu_{\bY}(\bx_i;\bbeta_g))^\top\bSigma^{-1}_{\bY|g}\balpha_{\bY|g}\notag\\
&-\frac{1}{2}\sum_{i=1}^n\sum_{g=1}^Gz_{ig}^{(r)}u_{ig}^{(r)}\mathrm{\tilde{E}}_{1_{\bY|ig}}^{(r)}\balpha^\top_{\bY|g}\bSigma^{-1}_{\bY|g}\balpha_{\bY|g}-\sum_{i=1}^n\sum_{g=1}^Gz_{ig}^{(r)}u_{ig}^{(r)}\mathrm{\tilde{E}}_{1_{\bY|ig}}^{(r)}\label{eq Q5}.
\end{align}

To do this we need to calculate $\mathrm{E}_{\btheta^{(r)}}\left(Z_{ig}|\bx_i,\by_i\right)$, $\mathrm{E}_{\btheta^{(r)}}\left(V_{ig}|\bx_i,\bz_i\right)$,  $\mathrm{E}_{\btheta^{(r)}}\left(U_{ig}|\by_i,\bz_i\right)$, $\mathrm{E}_{1_{\bX|ig}}$, $\mathrm{E}_{2_{\bX|ig}}$, $\mathrm{\tilde{E}}_{1_{\bX|ig}}$, $\mathrm{\tilde{E}}_{2_{\bX|ig}}$, 
 $\mathrm{E}_{1_{\bY|ig}}$, $\mathrm{E}_{2_{\bY|ig}}$, $\mathrm{\tilde{E}}_{1_{\bY|ig}}$, $\mathrm{\tilde{E}}_{2_{\bY|ig}}$,
 $i=1,\dots, n$ and $g=1,\dots,G$. Note that the terms on the $Q$-function where $\mathrm{E}_{3_{\bX|ig}}$, $\mathrm{\tilde{E}}_{3_{\bX|ig}}$,
$\mathrm{E}_{3_{\bY|ig}}$, 
and $\mathrm{\tilde{E}}_{3_{\bY|ig}}$ are constant with respect to the model parameters $\btheta$; therefore, they are not required in our calculations but we provide them nevertheless for completeness. The updates are calculated as
\begin{align}
    \mathrm{E}_{\btheta^{(r)}}\left(Z_{ig}|\bx_i,\by_i\right)
    &=\frac{\pi_g^{(r)}f_{\text{cSAL}}\left(\by_i;\bmu_{\bY}\left(\bx_i;\bbeta_g^{(r)}\right),\bSigma_{\bY|g}^{(r)},\balpha_{\bY|g}^{(r)},\delta_{\bY|g}^{(r)},\eta_{\bY|g}^{(r)}\right)f_{\text{cSAL}}\left(\bx_i;\bmu_{\bX|g}^{(r)},\bSigma_{\bX|g}^{(r)},\balpha_{\bX|g}^{(r)},\delta_{\bX|g}^{(r)},\eta_{\bX|g}^{(r)}\right)}{p\left(\bx_i,\by_i;\btheta^{(r)}\right)}\notag\\
    &:=z_{ig}^{(r)},\\
  \mathrm{E}_{\btheta^{(r)}}\left(V_{ig}|\bx_i,\bz_i\right)&=\frac{\left(1-\delta_{\bX|g}^{(r)}\right)f_{\text{SAL}}\left(\bx_i;\bmu_{\bX|g}^{(r)},\bSigma_{\bX|g}^{(r)},\balpha_{\bX|g}^{(r)}\right)}{f_{\text{cSAL}}\left(\bx_i;\bmu_{\bX|g}^{(r)},\bSigma_{\bX|g}^{(r)},\balpha_{\bX|g}^{(r)},\delta_{\bX|g}^{(r)},\eta_{\bX|g}^{(r)}\right)} := v_{ig}^{(r)},\\
    \mathrm{E}_{\btheta^{(r)}}\left(U_{ig}|\by_i,\bz_i\right)&=\frac{\left(1-\delta_{\bY|g}^{(r)}\right)f_{\text{cSAL}}\left(\by_i;\bmu_{\bY}\left(\bx_i;\bbeta_g^{(r)}\right),\bSigma_{\bY|g}^{(r)},\balpha_{\bY|g}^{(r)}\right)}{f_{\text{cSAL}}\left(\by_i;\bmu_{\bY}\left(\bx_i;\bbeta_g^{(r)}\right),\bSigma_{\bY|g}^{(r)},\balpha_{\bY|g}^{(r)},\delta_{\bY|g}^{(r)},\eta_{\bY|g}^{(r)}\right)}:=u_{ig}^{(r)}.
\end{align}
Since $W_{ig}|\bx_i,z_{ig}=1\sim \mathcal{GIG}(a_g,b_{ig},v)$ and $\tilde{W}_{ig}|\bx_i,z_{ig}=1\sim \mathcal{GIG}(a_g,c_{ig},v)$ where $\mathcal{GIG}(a,b,v)$ denotes the generalized inverse Gaussian (GIG) distribution with parameters  $a>0,b>0$ and $v\in \mathcal{R}$, $a_{\bX|g}^{(r)}=2+\left(\balpha_{\bX|g}^{(r)}\right)^\top\left(\bSigma_{\bX|g}^{(r)}\right)^{-1}\balpha_{\bX|g}^{(r)}$, $b_{\bX|ig}^{(r)}=\left(\bx_i-\bmu_{\bX|g}^{(r)}\right)^\top\left(\bSigma_{\bX|g}^{(r)}\right)^{-1}\left(\bx_i-\bmu_{\bX|g}^{(r)}\right)$
and \\$c_{\bX|ig}^{(r)}=\left(\bx_i-\bmu_{\bX|g}^{(r)}\right)^\top\left(\eta_{\bX|g}^{(r)}\bSigma_{\bX|g}^{(r)}\right)^{-1}\left(\bx_i-\bmu_{\bX|g}^{(r)}\right)$ and ${v_{\bX}=(2-d_{\bX})/2}$, it follows that 
\begin{align}
    \mathrm{E}^{(r)}_{1_{\bX|ig}}&:=\mathrm{E}\left(W_{ig}|\bx_i, \btheta^{(r)}\right)=\frac{\sqrt{b_{\bX|ig}^{(r)}}K_{v_{\bX}+1}\left(\sqrt{a_{\bX|g}^{(r)}b_{\bX|ig}^{(r)}}\right)}{\sqrt{a_{\bX|g}^{(r)}}K_{v_{\bX}}\left(\sqrt{a_{\bX|g}^{(r)}b_{\bX|ig}^{(r)}}\right)},\notag\\
\mathrm{E}^{(r)}_{2_{\bX|ig}}&:=\mathrm{E}\left(W_{ig}^{-1}|\bx_i, \btheta^{(r)}\right)=\frac{\sqrt{a_{\bX|g}^{(r)}}K_{v_{\bX}+1}\left(\sqrt{a_{\bX|g}^{(r)}b_{\bX|ig}^{(r)}}\right)}{\sqrt{b_{\bX|ig}^{(r)}}K_{v_{\bX}}\left(\sqrt{a_{\bX|g}^{(r)}b_{\bX|ig}^{(r)}}\right)}-\frac{2v_{\bX}}{b_{\bX|ig}^{(r)}},\notag\\
    \mathrm{E}^{(r)}_{3_{\bX|ig}}&:=\mathrm{E}\left(\mathrm{ln} W_{ig}|\bx_i, \btheta^{(r)}\right)=\mathrm{ln}\frac{\sqrt{b_{\bX|ig}^{(r)}}}{\sqrt{a_{\bX|g}^{(r)}}}+\frac{\partial}{\partial v}K_{v_{\bX}}\left(\sqrt{a_{\bX|g}^{(r)}b_{\bX|ig}^{(r)}}\right),\notag\\
\tilde{\mathrm{E}}^{(r)}_{1_{\bX|ig}}&:=\mathrm{E}\left(\tilde{W}_{ig}|\bx_i, \btheta^{(r)}\right)=\frac{\sqrt{c_{\bX|ig}^{(r)}}K_{v_{\bX}+1}\left(\sqrt{a_{\bX|g}^{(r)}c_{\bX|ig}^{(r)}}\right)}{\sqrt{a_{\bX|g}^{(r)}}K_{v_{\bX}}\left(\sqrt{a_{\bX|g}^{(r)}c_{\bX|ig}^{(r)}}\right)},\notag\\
\tilde{\mathrm{E}}^{(r)}_{2_{\bX|ig}}&:=\mathrm{E}\left(\tilde{W}_{ig}^{-1}|\bx_i, \btheta^{(r)}\right)=\frac{\sqrt{a_{\bX|g}^{(r)}}K_{v_{\bX}+1}\left(\sqrt{a_{\bX|g}^{(r)}c_{\bX|ig}^{(r)}}\right)}{\sqrt{c_{\bX|ig}^{(r)}}K_{v_{\bX}}\left(\sqrt{a_{\bX|g}^{(r)}c_{\bX|ig}^{(r)}}\right)}-\frac{2v_{\bX}}{c_{\bX|ig}^{(r)}},\notag\\
    \tilde{\mathrm{E}}^{(r)}_{3_{\bX|ig}}&:=\mathrm{E}\left(\mathrm{ln} \tilde{W}_{ig}|\bx_i, \btheta^{(r)}\right)=\mathrm{ln}\frac{\sqrt{c_{\bX|ig}^{(r)}}}{\sqrt{a_{\bX|g}^{(r)}}}+\frac{\partial}{\partial v}K_{v_{\bX}}\left(\sqrt{a_g^{(r)}c_{ig}^{(r)}}\right)\notag.
\end{align}
Similarly, since $W_{ig}|\by_i,z_{ig}=1\sim \mathcal{GIG}(a_g,b_{ig},v)$ and $\tilde{W}_{ig}|\by_i,z_{ig}=1\sim \mathcal{GIG}(a_g,c_{ig},v)$ where $a>0,b>0$ and $v\in \mathcal{R}$, $a_{\bY|g}^{(r)}=2+\left(\balpha_{\bY|g}^{(r)}\right)^\top\left(\bSigma_{\bY|g}^{(r)}\right)^{-1}\balpha_{\bY|g}^{(r)}$, $b_{\bY|ig}^{(r)}=\left(\by_i-\bmu_{\bY}\left(\bx_i;\bbeta_g^{(r)}\right)\right)^\top\left(\bSigma_{\bY|g}^{(r)}\right)^{-1}\left(\by_i-\bmu_{\bY}\left(\bx_i;\bbeta_g^{(r)}\right)\right)$
and $c_{\bY|ig}^{(r)}=\left(\by_i-\bmu_{\bY}\left(\bx_i;\bbeta_g^{(r)}\right)\right)^\top \left(\eta_{\bY|g}^{(r)}\bSigma_{\bY|g}^{(r)}\right)^{-1}\left(\by_i-\bmu_{\bY}\left(\bx_i;\bbeta_g^{(r)}\right)\right)$ and ${v_{\bY}=(2-d_{\bY})/2}$, it follows that 
\begin{align*}
    \mathrm{E}^{(r)}_{1_{\bY|ig}}&:=\mathrm{E}\left(W_{ig}|\by_i, \btheta^{(r)}\right)=\frac{\sqrt{b_{\bY|ig}^{(r)}}K_{v_{\bY}+1}\left(\sqrt{a_{\bY|g}^{(r)}b_{\bY|ig}^{(r)}}\right)}{\sqrt{a_{\bY|g}^{(r)}}K_{v_{\bY}}\left(\sqrt{a_{\bY|g}^{(r)}b_{\bY|ig}^{(r)}}\right)},\\
\mathrm{E}^{(r)}_{2_{\bY|ig}}&:=\mathrm{E}\left(W_{ig}^{-1}|\by_i, \btheta^{(r)}\right)=\frac{\sqrt{a_{\bY|g}^{(r)}}K_{v_{\bY}+1}\left(\sqrt{a_{\bY|g}^{(r)}b_{\bY|ig}^{(r)}}\right)}{\sqrt{b_{\bY|ig}^{(r)}}K_{v_{\bY}}\left(\sqrt{a_{\bY|g}^{(r)}b_{\bY|ig}^{(r)}}\right)}-\frac{2v_{\bY}}{b_{\bY|ig}^{(r)}},\\
    \mathrm{E}^{(r)}_{3_{\bY|ig}}&:=\mathrm{E}\left(\mathrm{ln} W_{ig}|\by_i, \btheta^{(r)}\right)=\mathrm{ln}\frac{\sqrt{b_{\bY|ig}^{(r)}}}{\sqrt{a_{\bY|g}^{(r)}}}+\frac{\partial}{\partial v}K_{v_{\bY}}\left(\sqrt{a_{\bY|g}^{(r)}b_{\bY|ig}^{(r)}}\right),\\    \tilde{\mathrm{E}}^{(r)}_{1_{\bY|ig}}&:=\mathrm{E}\left(\tilde{W}_{ig}|\by_i, \btheta^{(r)}\right)=\frac{\sqrt{c_{\bY|ig}^{(r)}}K_{v_{\bY}+1}\left(\sqrt{a_{\bY|g}^{(r)}c_{\bY|ig}^{(r)}}\right)}{\sqrt{a_{\bY|g}^{(r)}}K_{v_{\bY}}\left(\sqrt{a_{\bY|g}^{(r)}c_{\bY|ig}^{(r)}}\right)},\\
\tilde{\mathrm{E}}^{(r)}_{2_{\bY|ig}}&:=\mathrm{E}\left(\tilde{W}_{ig}^{-1}|\by_i, \btheta^{(r)}\right)=\frac{\sqrt{a_{\bY|g}^{(r)}}K_{v_{\bY}+1}\left(\sqrt{a_{\bY|g}^{(r)}c_{\bY|ig}^{(r)}}\right)}{\sqrt{c_{\bY|ig}^{(r)}}K_{v_{\bY}}\left(\sqrt{a_{\bY|g}^{(r)}c_{\bY|ig}^{(r)}}\right)}-\frac{2v_{\bY}}{c_{\bY|ig}^{(r)}},\\
\tilde{\mathrm{E}}^{(r)}_{3_{\bY|ig}}&:=\mathrm{E}\left(\mathrm{ln} \tilde{W}_{ig}|\by_i, \btheta^{(r)}\right)=\mathrm{ln}\frac{\sqrt{c_{\bY|ig}^{(r)}}}{\sqrt{a_{\bY|g}^{(r)}}}+\frac{\partial}{\partial v}K_{v_{\bY}}\left(\sqrt{a_g^{(r)}c_{ig}^{(r)}}\right).
\end{align*}
Then, by substituting $z_{ig}$ with $z_{ig}^{(r)}$, $v_{ig}$ with $v_{ig}^{(r)}$, $u_{ig}$ with $u_{ig}^{(r)}$, $\mathrm{E}_{1_{\bX|ig}}$ with $\mathrm{E}_{1_{\bX|ig}}^{(r)}$,
$\mathrm{E}_{2_{\bX|ig}}$ with $\mathrm{E}_{2_{\bX|ig}}^{(r)}$,
$\mathrm{\tilde{E}}_{1_{\bX|ig}}$ with $\mathrm{\tilde{E}}_{1_{\bX|ig}}^{(r)}$, $\mathrm{\tilde{E}}_{1_{\bX|ig}}$, with $\mathrm{\tilde{E}}_{1_{\bX|ig}}^{(r)}$ 
 $\mathrm{E}_{1_{\bY|ig}}$ with  $\mathrm{E}_{1_{\bY|ig}}^{(r)}$,
 $\mathrm{E}_{2_{\bY|ig}}$ with $\mathrm{E}_{2_{\bY|ig}}^{(r)}$,
 $\mathrm{\tilde{E}}_{1_{\bY|ig}}$ with $\mathrm{\tilde{E}}_{1_{\bY|ig}}^{(r)}$, and $\mathrm{\tilde{E}}_{2_{\bY|ig}}$ with $\mathrm{\tilde{E}}_{2_{\bY|ig}}^{(r)}$ in \eqref{eq complete log likelihood}, we obtain 
$Q\left(\btheta|\btheta^{(r)}\right)$.
\subsubsection*{CM-step 1}
The first CM-step on the $(r+1)$-th iteration of the ECM algorithm requires the calculation of $\btheta_1^{(r+1)}$ as the value of $\btheta_1$ that maximizes $Q\left(\btheta|\btheta^{(r)}\right)$.
In particular, after some algebra shown in the Appendix, we obtain
\begin{align*}
    \pi_g^{(r+1)}=&\frac{n_g^{(r)}}{n},\\
	\delta_{\bX|g}^{(r+1)}=&\frac{1}{n_g^{(r)}}\sum_{i=1}^{n}z_{ig}^{(r)}v_{ig}^{(r)},\\
\bmu_{\bX|g}^{(r+1)}=&\frac{\sum\limits^n_{i=1}B_i^{(r)}\left[ \sum \limits^{n}_{i=1}A_i^{(r)}\bx_i\right]-\sum\limits^n_{i=1}C_i^{(r)}\left[ \sum \limits^{n}_{i=1}C_i^{(r)}\bx_i\right]}{\sum\limits^n_{i=1}B_i^{(r)}\sum\limits^n_{i=1}A_i^{(r)}-\left(\sum\limits^n_{i=1}C_i^{(r)}\right)^2},\\
\balpha_{\bX|g}^{(r+1)}=&\frac{ \sum \limits^{n}_{i=1}A_i^{(r)}\left[  \sum \limits^{n}_{i=1}C_i^{(r)}\bx_i\right]- \sum \limits^{n}_{i=1}C_i^{(r)}\left[ \sum \limits^{n}_{i=1}A_i^{(r)}\bx_i\right]}{ \sum \limits^{n}_{i=1}B_i^{(r)} \sum \limits^{n}_{i=1}A_i^{(r)}-\left( \sum \limits^{n}_{i=1}C_i^{(r)}\right)^2},\\
    \bSigma_{\bX|g}^{(r+1)}=&\frac{1}{n_{g}^{(r)}}\sum^n_{i=1}A_i^{(r)}\left(\bx_i-\bmu_{\bX|g}^{(r+1)}\right)\left(\bx_i-\bmu_{\bX|g}^{(r+1)}\right)^\top-\frac{2}{n_g^{(r)}}\sum^{n}_{i=1}C_i^{(r)}\left(\bx_i-\bmu_{\bX|g}^{(r+1)}\right)\left(\balpha_{\bX|g}^{(r+1)}\right)^\top\\
     &+\frac{1}{n_g^{(r)}}\sum^{n}_{i=1}B_i^{(r)}\balpha_{\bX|g}^{(r+1)}\left(\balpha_{\bX|g}^{(r+1)}\right)^\top,\\
	\delta_{\bY|g}^{(r+1)}=&\frac{1}{n_g^{(r)}}\sum_{i=1}^{n}z_{ig}^{(r)}u_{ig}^{(r)},\\
    \bbeta_g^{(r+1)} =& \left(\left(
    \sum\limits^n_{k=1}D_i\bx_i^*{\bx_i^*}^\top-\frac{\sum\limits^n_{i=1}G_i\left(\sum\limits^n_{k=1}G_k\bx_k^*\right){\bx_i^*}^\top}{\sum\limits^n_{k=1}F_k}\right)^\top\right)^{-1}\left( \sum\limits^n_{k=1}D_i\by_i{\bx_i^*}^\top-\frac{\sum\limits^n_{i=1}G_i\left(\sum\limits^n_{k=1}G_k\by_k\right){\bx_i^*}^\top}{\sum\limits^n_{k=1}F_k}\right)^\top,\\
\balpha_{\bY|g}^{(r+1)}=&\frac{ \sum \limits^{n}_{i=1}G_i^{(r)}\left(\by_i-\bmu_Y\left(\bx_i;\bbeta_g^{(r+1)}\right)\right)}{ \sum \limits^{n}_{i=1}F_i^{(r)}},\\
    \bSigma_{\bY|g}^{(r+1)}=&\frac{1}{n_{g}^{(r)}}\sum^n_{i=1}D_i^{(r)}\left(\by_i-\bmu_{Y}\left(\bx_i;\bbeta_g^{(r+1)}\right)\right)\left(\by_i-\bmu_{Y}\left(\bx_i;\bbeta_g^{(r+1)}\right)\right)^\top\\
    &-\frac{2}{n_g^{(r)}}\sum^{n}_{i=1}F_i^{(r)}\left(\by_i-\bmu_{Y}\left(\bx_i;\bbeta_g^{(r+1)}\right)\right)\left(\balpha_{\bY|g}^{(r+1)}\right)^\top\\
     &+\frac{1}{n_g^{(r)}}\sum^{n}_{i=1}G_i^{(r)}\balpha_{\bY|g}^{(r+1)}\left(\balpha_{\bY|g}^{(r+1)}\right)^\top,
\end{align*}
where
\begin{align}
    A_i^{(r)}&=z_{ig}^{(r)}\left(\left(1-v_{ig}^{(r)}\right)\mathrm{E}_{2_{\bX|ig}}^{(r)}+\frac{v_{ig}^{(r)}}{\eta_{\bX|g}^{(r)}}\mathrm{\tilde{E}}_{2_{\bX|ig}}^{(r)}\right),\notag\\
    B_i^{(r)}&=z_{ig}^{(r)}\left(\left(1-v_{ig}^{(r)}\right)\mathrm{E}_{1_{\bX|ig}}^{(r)}+v_{ig}^{(r)}\mathrm{\tilde{E}}_{1_{\bX|ig}}^{(r)}\right),\notag\\
    C_i^{(r)}&= z_{ig}^{(r)}\left(\left(1-v_{ig}^{(r)}\right)+\frac{v_{ig}^{(r)}}{\sqrt{\eta_{\bX|g}^{(r)}}}\right),\notag\\
    D_i^{(r)}&=z_{ig}^{(r)}\left(\left(1-u_{ig}^{(r)}\right)\mathrm{E}_{2_{\bY|ig}}^{(r)}+\frac{u_{ig}^{(r)}}{\eta_{\bY|g}^{(r)}}\mathrm{\tilde{E}}_{2_{\bY|ig}}^{(r)}\right),\label{eq notation D}\\
    F_i^{(r)}&=z_{ig}^{(r)}\left(\left(1-u_{ig}^{(r)}\right)\mathrm{E}_{1_{\bY|ig}}^{(r)}+u_{ig}^{(r)}\mathrm{\tilde{E}}_{1_{\bY|ig}}^{(r)}\right),\label{eq notation F}\\
    G_i^{(r)}&= z_{ig}^{(r)}\left(\left(1-u_{ig}^{(r)}\right)+\frac{u_{ig}^{(r)}}{\sqrt{\eta_{\bY|g}^{(r)}}}\right).\label{eq notation G}
\end{align}

\subsubsection*{CM-step 2}
In the second CM-step, on the $(r+1)$-th iteration of the ECM algorithm, we calculate $\btheta_2^{(r+1)}$ as the value of $\btheta_2$ that maximizes $Q\left(\btheta|\btheta_1=\btheta^{(r+1)}\right)$.
More specifically, an update for $\eta_{\bX|g}$ is calculated by differentiating $Q_3^{\text{bad}}\left(\bmu_{\bX},\bSigma_{\bX},\balpha_{\bX},\boldeta_{\bX}\right)$ in \eqref{eq Q3 bad} with respect to $\eta_{\bX|g}$. 
This gives
\begin{align}\label{eq Q3 bad partial derivative}
    \frac{\partial Q_3^{\text{bad}}\left(\bmu_{\bX},\bSigma_{\bX},\balpha_{\bX},\boldeta_{\bX}\right)}{\partial \eta_{\bX|g}}=&-\frac{d_{\bX}}{2\eta_{\bX|g}}\sum_{i=1}^nz_{ig}v_{ig}+\frac{1}{2\eta_{\bX|g}^2}\sum_{i=1}^nz_{ig}v_{ig}\mathrm{\tilde{E}}_{2_{\bX|ig}}(\bx_i-\bmu_{\bX|g})^\top\bSigma_{\bX|g}^{-1}(\bx_i-\bmu_{\bX|g})\notag\\
    &-\frac{1}{2\eta_{\bX|g}^{3/2}}\sum_{i=1}^nz_{ig}v_{ig}(\bx_i-\bmu_{\bX|g})^\top\bSigma_{\bX|g}^{-1}\balpha_{\bX|g}.
\end{align}
Now, by setting the partial derivative in \eqref{eq Q3 bad partial derivative} equal to zero and multiplying by $-2\eta_{\bX|g}^2$, allows us to express this equation as
\begin{align*}
    a_{\bX|g}^*\eta_{\bX|g}+b_{\bX|j}^*\sqrt{\eta_{\bX|g}}+c_{\bX|j}^*=0,
\end{align*}
where 
\begin{align*}
    a_{\bX|g}^{*(r+1)}=d_{\bX}\sum_{i=1}^nz_{ig}^{(r)}v_{ig}^{(r)}, \quad
    b_{\bX|g}^{*(r+1)}=\sum_{i=1}^nz_{ig}^{(r)}v_{ig}^{(r)}\left(\bx_i-\bmu_{\bX|g}^{(r)}\right)^\top{\left(\bSigma_{\bX|g}^{(r)}\right)}^{-1}\balpha_{\bX|g}^{(r)},
\end{align*}
\begin{align*}
c_{\bX|g}^{*(r+1)}=-\sum_{i=1}^nz_{ig}^{(r)}v_{ig}^{(r)}\mathrm{\tilde{E}}^{(r)}_{2_{\bX|ig}}\left(\bx_i-\bmu_{\bX|g}^{(r)}\right)^\top\left(\bSigma_{\bX|g}^{(r)}\right)^{-1}\left(\bx_i-\bmu_{\bX|g}^{(r)}\right).
\end{align*}
Using the quadratic formula, we can calculate the real positive root for this equation as
\begin{align*}
    \eta_{\bX|g}^*=\left(\frac{-b_{\bX|g}^{*(r+1)}+\sqrt{\left(b_{\bX|g}^{*(r+1)}\right)^2-4a_{\bX|g}^{*(r+1)}c_{\bX|g}^{*(r+1)}}}{2a_{\bX|g}^{*(r+1)}}\right)^2,
\end{align*}
where a lower limit of 1 is set to ensure the integrity of the parameter. Hence the update for $\eta_{\bX|g}$ is given by
\begin{align*}
    \eta_{\bX|g}^{(r+1)}=\mathrm{max}\left\{1,\eta_{\bX|g}^*\right\}.
\end{align*}
Similarly, an update $\eta_{\bY|g}^{(r+1)}$ for $\eta_{\bY|g}$ is calculated by differentiating $Q_5^{\text{bad}}\left(\bbeta,\bSigma_{\bY},\balpha_{\bY},\boldeta_{\bY}\right)$ with respect to  $\eta_{\bY|g}$, which results in 
\begin{align*}
    \eta_{\bY|g}^*=\left(\frac{-b_{\bY|g}^{*(r+1)}+\sqrt{\left(b_{\bY|g}^{*(r+1)}\right)^2-4a_{\bY|g}^{*(r+1)}c_{\bY|g}^{*(r+1)}}}{2a_{\bY|g}^{*(r+1)}}\right)^2,
\end{align*}
where 
\begin{align*}
    a_{\bY|g}^{*(r+1)}=d_{\bY}\sum_{i=1}^nz_{ig}^{(r)}u_{ig}^{(r)}, \quad
    b_{\bY|g}^{*(r+1)}=\sum_{i=1}^nz_{ig}^{(r)}u_{ig}^{(r)}\left(\by_i-\bmu_{Y}\left(\bx_i;\bbeta_g^{(r+1)}\right)\right)^\top{\left(\bSigma_{\bY|g}^{(r)}\right)}^{-1}\balpha_{\bY|g}^{(r)},
\end{align*}
\begin{align*}
c_{\bY|g}^{*(r+1)}=-\sum_{i=1}^nz_{ig}^{(r)}u_{ig}^{(r)}\mathrm{\tilde{E}}^{(r)}_{2_{\bY|ig}}\left(\by_i-\bmu_{Y}\left(\bx_i;\bbeta_g^{(r+1)}\right)\right)^\top\left(\bSigma_{\bY|g}^{(r)}\right)^{-1}\left(\by_i-\bmu_{Y}\left(\bx_i;\bbeta_g^{(r+1)}\right)\right),
\end{align*}
and hence an update for $\eta_{\bY|g}$ is given by
\begin{align}
    \eta_{\bY|g}^{(r+1)}=\mathrm{max}\left\{1,\eta_{\bY|g}^*\right\}.
\end{align}
\subsection{Notes on an EM algorithm for the SALCWM}\label{Section EM algorithm for SALCWM}
Let $(\bx_1,\by_1),\dots,(\bx_n,\by_n)$ instead be an observed sample from the SALCWM \eqref{pdf asymmetric Laplace cwm}. 
The ML estimates for parameters of the SALCWM follow easily from the ECM described in Section \ref{section ecm algorithm for cSALCWM}. 
In this case, we employ an EM algorithm instead of an ECM algorithm, as there is no need to estimate the contamination parameters (i.e., $\delta_{\bX|g}$, $\eta_{\bX|g}$, $\delta_{\bY|g}$, and $\eta_{\bY|g}$). 
Consequently, the EM algorithm only consists of a single M-step, instead of the two CM-steps required in the ECM algorithm.

In this setting, there is only one source of incompleteness, which arises from the uncertainty regarding the component membership of each observation. 
This uncertainty is represented by an indicator vector $\bz_i = (z_{i1}, \dots, z_{iG})$, where $z_{ig} = 1$ if the observation $(\bx_i, \by_i)$ belongs to component $g$, and $z_{ig} = 0$ otherwise, as in the case of the SALCWM. 
Since we do not account for additional sources of uncertainty in this case of whether a point $(\bx_i, \by_i)$ is an outlier or a leverage point, we set $\bv_i$ and $\bu_i$ in the ECM algorithm equal to $\boldsymbol{0}$.

Thus, in the M-step on the $(r+1)-th$ iteration, the EM algorithm updates for $\pi_g$, $\bmu_{\bX|g}$, $\balpha_{\bX|g}$, $\bSigma_{\bX|g}$, $\bbeta_g$ , $\balpha_{\bY|g}$, and $\bSigma_{\bY|g}$ are calculated as described in CM-step 1 in Section \ref{section ecm algorithm for cSALCWM}, but instead with

\begin{align*}
    A_i^{(r)}=z_{ig}^{(r)}\mathrm{E}_{2_{\bX|ig}}^{(r)}, \quad
    B_i^{(r)}=z_{ig}^{(r)}\mathrm{E}_{1_{\bX|ig}}^{(r)}, \quad
    C_i^{(r)}= z_{ig}^{(r)},\\
    D_i^{(r)}=z_{ig}^{(r)}\mathrm{E}_{2_{\bY|ig}}^{(r)},\quad
    F_i^{(r)}=z_{ig}^{(r)}\mathrm{E}_{1_{\bY|ig}}^{(r)},\quad
    G_i^{(r)}= z_{ig}^{(r)}.
\end{align*}

\subsection{Initialization}
The choice of the starting values for EM-based algorithms constitutes an important issue \citep{biernacki2003choosing,punzo2017robust}
In the first iteration of the EM algorithm for the SALCWM,
we follow the initialization approach used in the $G$-component \texttt{MixSAL} package \citep{franczak2018package}, where initial cluster labels are assigned using k-means clustering. 
$\bmu_{\bX|g}^{(0)}$ is then computed as the weighted mean of $\bX$ of the points assigned to the $g$-th cluster, while $\bSigma_{\bX|g}^{(0)}$ is the weighted covariance matrix. 
The asymmetry parameter, $\balpha_{\bX|g}^{(0)}$, is initially set to $\boldsymbol{0}$, ensuring a symmetric starting point before the EM algorithm iteratively adjusts it to account for skewness. 
Similarly,  $\bmu_{\bY|g}^{(0)}$ is computed as the weighted mean of $\bY$ of the points assigned to the $g$-th cluster, while $\bSigma_{\bY|g}^{(0)}$ is the weighted covariance matrix.  
The mixing proportions are initialized as the mean of the cluster assignment matrix. 

For the $G$-component cSALCWM, we adopt the approach used by \citet{punzo2017robust}. 
The $G$-component SALCWM in \eqref{pdf asymmetric Laplace cwm} can be seen as nested in the $G$-component cSALCWM in \eqref{pdf contaminated asymmetric Laplace cwm} when $\delta_{\bX|g},\delta_{\bY|g}\to0^+$ or $\delta_{\bX|g},\delta_{\bY|g}\to1^-$, and $\eta_{\bX|g},\eta_{\bY|g}\to1^+$, for $g=1,\dots,G$. 
Under these conditions, $u_{ig}, v_{ig}\to 0^+$, $i=1,\dots,n$ and $g=1,\dots,G$, causing model \eqref{pdf contaminated asymmetric Laplace cwm} to tend to \eqref{pdf asymmetric Laplace cwm}. 
Consequently, the posterior probabilities, say $z_{ig}^{(0)}$, obtained from the EM algorithm for the SALCWM -- along with the constraints $u_{ig}^{(0)},v_{ig}^{(0)}\to0^+$, and $\eta_{\bX|g}^{(0)},\eta_{\bX|g}^{(0)}\to1^+$ -- can be used to initialize the first CM-step of our ECM algorithm in Section \ref{section ecm algorithm for cSALCWM}. 
From an operational point of view, due to the monotonicity property of the ECM algorithm, this also guarantees that the observed-data log-likelihood of the cSALCWM will always be greater than, or equal to, the observed-data log-likelihood of the initial SALCWM. 
In our analysis, we set $u_{ig}^{(0)},v_{ig}^{(0)}=0.001$ and $\eta_{\bX|g}^{(0)},\eta_{\bY|g}^{(0)}=1.001$. These values are not set equal to 0 and 1, respectively, to avoid singularity issues within the first iteration \citep{punzo2020high}.
\subsection{Convergence}
The Aitken acceleration \citep{aitken1926iii} is used to assess the convergence of our EM algorithms described in Section \ref{section ecm algorithm for cSALCWM} and \ref{Section EM algorithm for SALCWM}. 
With this estimate, we can determine whether the algorithm has reached convergence, meaning the log-likelihood is sufficiently close to its estimated asymptotic value. 
The Aitken acceleration at iteration $ r+1 $ is calculated as  
\begin{align*}
a^{(r+1)} = \frac{l^{(r+2)} - l^{(r+1)}}{l^{(r+1)} - l^{(r)}},
\end{align*}
where $l^{(r)}$ represents the observed-data log-likelihood at iteration  $r$. Based on this, the asymptotic estimate of the log-likelihood at iteration $r+2$ is given by  
\begin{align*}
l^{(r+2)}_{\infty} = l^{(r+1)} + \frac{1}{1 - a^{(r+1)}} \left( l^{(r+2)} - l^{(r+1)} \right).
\end{align*}
Following \citet{bohning1994distribution}, the EM-based algorithms are considered to have converged when the difference  
\begin{align*}
    l^{(r+2)}_{\infty} - l^{(r+1)} < \epsilon
\end{align*}
is both positive and smaller than a predefined threshold  $\epsilon$. In our analysis, we set  $\epsilon = 0.00001$.

\subsection{Dealing with infinite log-likelihood values}
As documented in \citet{ morris2019asymmetric} and \cite{mclaughlin2024unsupervised}, complications can arise when estimating the location parameter $\bmu_{g}$ of a SAL and cSAL distribution, and consequently, the SALCWM and cSALCWM. Computational singularities occur when the parameter $\bmu_{\bX|g}$ or $\bmu_{\bY|g}\left(\bx_i;\bbeta_g\right)$ is equal, or very close, to some observation $\bx_{i}$ or $\by_i$, respectively, in the data. In our CWM family, such singularities manifest when calculating the expected values $\mathrm{E}_{2_{\bX|ig}}$ or $\mathrm{E}_{2_{\bY|ig}}$ (and $\tilde{\mathrm{E}}_{2_{\bX|ig}}$ or $\tilde{\mathrm{E}}_{2_{\bY|ig}}$ in the case of the cSALCWM). To remedy this issue, we stop updating $\bmu_{\bX|g}$ and $\bmu_{\bY|g}\left(\bx_i;\bbeta_g\right)$ if the Euclidean distance between $\bmu_{\bX|g}$ and $\bx_{i}$, or between $\bmu_{\bY|g}\left(\bx_i;\bbeta\right)$ and $\by_i$ for $i=1,\dots,n$ and $g=1,\dots,G$, is less than a user-specified value that is sufficiently small. In practice, we found that $10^{-5}$ was an effective value for our purposes. As noted by \citet{franczak2013mixtures}, this is a simple but effective solution to dealing with the problem at hand. Other methods for dealing with infinite likelihoods can be found in \cite{fang2023tackling}.

\subsection{Cluster-wise automatic detection of atypical points from the cSALCWM}
For an observation $(\bx_i,\by_i)$ from the cSALCWM in \eqref{pdf contaminated asymmetric Laplace cwm}, the classification involves
\begin{enumerate}
    \item determining its component of membership;
    \item establishing if it is typical, outlier, good leverage, or bad leverage in that component. 
\end{enumerate}
Let $\hat{\bu}_i$, $\hat{\bv}_i$ and $\hat{\bz}_i$ denote the expected values of $\bu_i$, $\bv_i$ and $\bz_i$, respectively, arising from the ECM algorithm at convergence. The component membership of $(\bx_i,\by_i)$ is determined using the maximum a posteriori probabilities (MAP) operator
\begin{align*}
    \text{MAP}(\hat{z}_{ig})=\begin{cases}
        1\quad \text{if max$_h\left\{\hat{z}_{ig}\right\}$ occurs in component $h=g$}, \\
        0 \quad \text{otherwise}.
    \end{cases}
\end{align*}
Next, we consider $\hat{u}_{ih}$ and $\hat{v}_{ih}$ where $h$ is selected such that $\text{MAP}(\hat{z}_{ih})=1$. These values provide key insights into whether $(\bx_i,\by_i)$ is an outlier or a leverage point, respectively, in group $h$. However, the user could be interested in obtaining a classification of this observation according to \tablename~\ref{table categorization of points in a regression analysis}. In such cases, the decision rule in \tablename~\ref{table rule of points in a regression analysis} could be applied. Thus, once the observation has been classified into one of the groups, applying the rule in \tablename~\ref{table rule of points in a regression analysis} provides deeper insights into its role within that group.
\begin{table}[!ht]
\caption{Rule for classifying a generic observation $(\bx_i,\by_i)$ into one of the four categories of \tablename~\ref{table categorization of points in a regression analysis}.}
\centering
\begin{tabular}{lrr}
\toprule
   \diagbox{$\hat{u}_{ih}$}{$\hat{v}_{ih}$} & \multicolumn{1}{c}{$[0,0.5)$} & \multicolumn{1}{c}{$[0.5,1]$} \\ \midrule
$[0,0.5)$ & Typical            & Good leverage                \\
$[0.5,1]$  & Outlier            & Bad leverage                \\ \bottomrule
\end{tabular}
\label{table rule of points in a regression analysis}
\end{table} 
\section{Sensitivity analysis}\label{Section sensitivity anslys}
In this study, we perform a sensitivity analysis to investigate the impact of atypical observations on the SALCWM and cSALCWM. 
To generate the data, we follow the simulation study performed by \citet{punzo2017robust} and we consider the following data generation processes with univariate conditional/marginal distributions (i.e., $d_{\bX}=d_{\bY}=1$) and $G=2$ mixture components:
\begin{enumerate}[(a)]
\item SALCWM with 1\% of points randomly substituted by outliers with coordinates $(\mu_{X|2},y^*)$, where $\mu_{X|2}$ is the mean of the second mixture component, and $y^*$ is generated from a uniform distribution over the interval $(8,10)$;
\item SALCWM with 1\% of points randomly substituted by good leverage points with coordinates $(x^*,y^*)$, where $x^*$ is generated from a uniform distribution over the interval $(8,10)$ and $y^*=\beta_{02}+\beta_{12}x^*$ lies on the straight line on the second mixture component;
    \item SALCWM with 1\% of points randomly substituted by bad leverage points with coordinates $(x^*,y^*)$ 
 where both $x^*$ and $y^*$ are generated from a uniform distribution over the interval $(8,10)$;
    \item SALCWM with 1\% of points randomly substituted by noise points generated from a uniform distribution over the interval $(-8,8)$.
\end{enumerate}
All of these data generation processes share the following parameters
\begin{align*}
    \pi=0.4,\quad 
    \bbeta_1=\begin{pmatrix}
        -2\\
        -0.2
    \end{pmatrix},
    \quad
    \bbeta_2=\begin{pmatrix}
        2\\
        0.2
    \end{pmatrix},
    \quad
    \Sigma_{Y|1}=\Sigma_{Y|1}=0.5,
      \quad
    \alpha_{Y|1}=-0.2, \quad \alpha_{Y|2}=0.2.
\end{align*}
The parameters relating to $X$ will be specified depending on either assignment independence or assignment dependence.
The four scenarios described cover various situations involving possible atypical values in real-world data, as outlined in \tablename~\ref{table categorization of points in a regression analysis}. Under each scenario, we generate 100 samples considering two different sample sizes, $n=250$ and $n=500$. The models are directly fitted with $G=2$. 
For comparison, we report the bias and the mean squared error (MSE) of the estimates for the mixture weight $\pi_1$ and all regression coefficients $\bbeta_1=(\beta_{01},\beta_{11})$ and $\bbeta_2=(\beta_{02},\beta_{12})$. Additionally, we assess the cSALCWM's ability to effectively classify atypical points according to \tablename~\ref{table rule of points in a regression analysis}. Specifically, we report the true positive rate (TPR), which measures the proportion of atypical observations that are correctly identified as atypical, and the false positive rate (FPR), corresponding to the proportion of typical points incorrectly classified as atypical.

\subsection{Comparison under assignment independence}
The first simulation study aims to compare the models under assignment independence. We do this by making use of Proposition \ref{proposition mix of regression special case of CWM} to generate data, which states that if the covariate parameters in the CWM are identical across all $G$ groups, it is a particular case of the CWMs which assumes assignment independence. In this case, all four data-generating processes share the following parameters
\begin{align*}
    \mu_{X|1}=\mu_{X|2}=0, \quad 
    \Sigma_{X|1}=\Sigma_{X|2}=1, \quad
    \alpha_{X|1}=\alpha_{X|2}=0.2.  
\end{align*}

\figurename~\ref{Fig assignment indep} illustrates a single replication of this setup. The black crosses represent the original SALCWM data points that were replaced, while the different atypical observations of each scenario under assignment independence are also depicted.

The obtained bias and MSE values are reported in \tablename~\ref{table bias mse values under assignment independence} for sample sizes $n=250$ and $n=500$, across 100 replications. We note that the bias and MSE values are small across all four scenarios. However, the contaminated models, cSALMRM and cSALCWM, generally have lower bias in absolute value and MSE than the SAL-based models, SALMRM and SALCWM. Under scenario (a), the bias and MSE values for our intercepts are more affected by outliers than slopes, which is natural if we recall that outliers, having coordinate $x=0$, have been added.
Additionally, regardless of the scenario, the MRMs and CWMs perform similar, which is expected since the two latent groups are generated under assignment independence, having the same distribution for $X$. Furthermore, both cSALMRM and cSALCWM effectively detect mild atypical values, achieving TPRs close to 1 and FPRs close to 0, as shown in Tables \ref{table tpr fpr under assignment independence for a to c} and \ref{table tpr fpr under assignment independence for d}. In scenario (d), the TPR is lower, which is not unexpected since some noisy points may resemble typical points, making them harder to detect as atypical; notably, the cSALCWM's TPR is nearly double that of the cSALMRM.

\begin{figure}[!ht]
    \centering
    \includegraphics[scale=0.6]{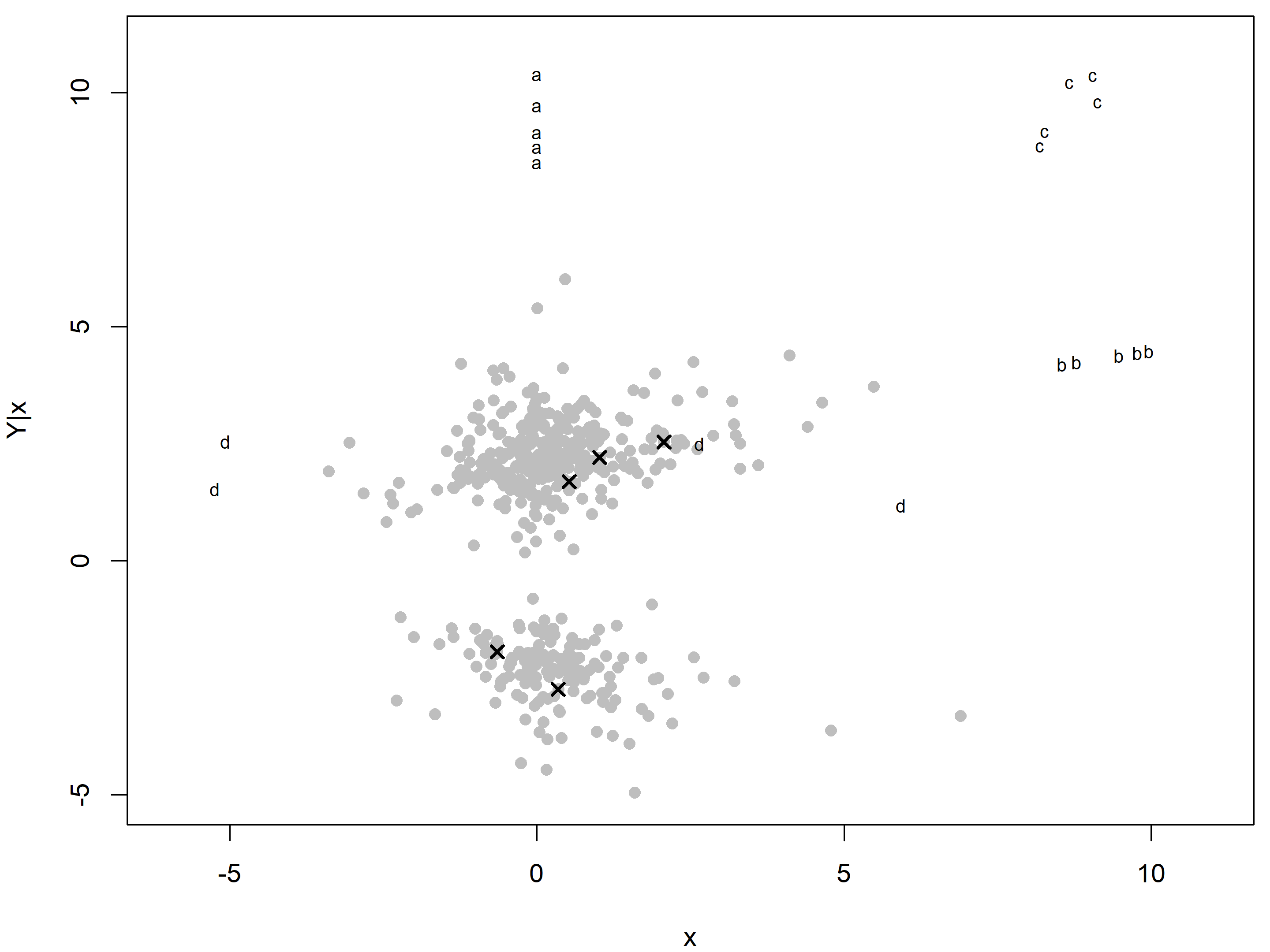}
    \caption{Example replication under assignment independence. The black crosses indicate original SALCWM data points that have been replaced, while the corresponding replacements are shown for each scenario.}
    \label{Fig assignment indep}
\end{figure}

\begin{table}[ht]
\centering
\caption{Bias and MSE values under assignment independence across 100 replications.}
\label{table bias mse values under assignment independence}
\begin{tabular}{lllrrlrrrrrlrr}
\toprule
                              &                       &              & \multicolumn{5}{c}{$n=250$}                                                                                & \multicolumn{1}{c}{} & \multicolumn{5}{c}{$n=500$}                                                                                \\ \cline{4-8} \cline{10-14} 
                              &                       &              & \multicolumn{2}{c}{MRM}                            &  & \multicolumn{2}{c}{CWM}                            & \multicolumn{1}{c}{} & \multicolumn{2}{c}{MRM}                            &  & \multicolumn{2}{c}{CWM}                            \\ \cline{4-5} \cline{7-8} \cline{10-11} \cline{13-14} 
                              &                       &              & \multicolumn{1}{c}{SAL} & \multicolumn{1}{c}{cSAL} &  & \multicolumn{1}{c}{SAL} & \multicolumn{1}{c}{cSAL} & \multicolumn{1}{c}{} & \multicolumn{1}{c}{SAL} & \multicolumn{1}{c}{cSAL} &  & \multicolumn{1}{c}{SAL} & \multicolumn{1}{c}{cSAL} \\ \midrule
\multirow{10}{*}{Scenario a)} & \multirow{5}{*}{Bias} & $\beta_{01}$ & -0.159                  & -0.097                   &  & -0.149                  & -0.095                   &                      & -0.219                  & -0.088                   &  & -0.233                  & -0.096                   \\
                              &                       & $\beta_{11}$ & -0.018                  & -0.009                   &  & -0.018                  & -0.011                   &                      & -0.003                  & 0.002                    &  & -0.005                  & 0.003                    \\
                              &                       & $\beta_{02}$ & -0.034                  & -0.019                   &  & -0.036                  & -0.018                   &                      & -0.031                  & -0.009                   &  & -0.030                  & -0.009                   \\
                              &                       & $\beta_{12}$ & -0.001                  & -0.002                   &  & -0.004                  & -0.002                   &                      & 0.002                   & 0.003                    &  & 0.002                   & 0.003                    \\
                              &                       & $\pi_1$      & 0.013                   & 0.000                    &  & 0.012                   & -0.001                   &                      & 0.026                   & 0.004                    &  & 0.028                   & 0.004                    \\
                              & \multirow{5}{*}{MSE}  & $\beta_{01}$ & 0.063                   & 0.039                    &  & 0.059                   & 0.038                    &                      & 0.090                   & 0.021                    &  & 0.098                   & 0.024                    \\
                              &                       & $\beta_{11}$ & 0.009                   & 0.008                    &  & 0.009                   & 0.007                    &                      & 0.005                   & 0.002                    &  & 0.006                   & 0.003                    \\
                              &                       & $\beta_{02}$ & 0.007                   & 0.006                    &  & 0.007                   & 0.006                    &                      & 0.004                   & 0.002                    &  & 0.004                   & 0.002                    \\
                              &                       & $\beta_{12}$ & 0.002                   & 0.002                    &  & 0.002                   & 0.002                    &                      & 0.001                   & 0.001                    &  & 0.001                   & 0.001                    \\
                              &                       & $\pi_1$      & 0.001                   & 0.000                    &  & 0.001                   & 0.000                    &                      & 0.002                   & 0.000                    &  & 0.002                   & 0.000                    \\
                              &                       &              &                         &                          &  &                         &                          &                      &                         &                          &  &                         &                          \\
\multirow{10}{*}{Scenario b)} & \multirow{5}{*}{Bias} & $\beta_{01}$ & -0.011                  & -0.010                   &  & -0.024                  & -0.023                   &                      & -0.010                  & -0.012                   &  & -0.024                  & -0.027                   \\
                              &                       & $\beta_{11}$ & -0.012                  & -0.012                   &  & -0.011                  & -0.013                   &                      & -0.003                  & -0.003                   &  & 0.000                   & 0.000                    \\
                              &                       & $\beta_{02}$ & 0.009                   & 0.007                    &  & 0.007                   & 0.006                    &                      & 0.004                   & 0.003                    &  & 0.003                   & 0.003                    \\
                              &                       & $\beta_{12}$ & -0.003                  & -0.002                   &  & -0.003                  & -0.002                   &                      & -0.001                  & 0.000                    &  & -0.001                  & 0.000                    \\
                              &                       & $\pi_1$      & -0.003                  & -0.003                   &  & -0.003                  & -0.003                   &                      & -0.003                  & -0.003                   &  & -0.003                  & -0.002                   \\
                              & \multirow{5}{*}{MSE}  & $\beta_{01}$ & 0.007                   & 0.008                    &  & 0.008                   & 0.008                    &                      & 0.003                   & 0.003                    &  & 0.004                   & 0.005                    \\
                              &                       & $\beta_{11}$ & 0.006                   & 0.006                    &  & 0.006                   & 0.006                    &                      & 0.002                   & 0.002                    &  & 0.002                   & 0.002                    \\
                              &                       & $\beta_{02}$ & 0.004                   & 0.004                    &  & 0.004                   & 0.004                    &                      & 0.002                   & 0.002                    &  & 0.002                   & 0.002                    \\
                              &                       & $\beta_{12}$ & 0.000                   & 0.000                    &  & 0.000                   & 0.000                    &                      & 0.000                   & 0.000                    &  & 0.000                   & 0.000                    \\
                              &                       & $\pi_1$      & 0.000                   & 0.000                    &  & 0.000                   & 0.000                    &                      & 0.000                   & 0.000                    &  & 0.000                   & 0.000                    \\
                              &                       &              &                         &                          &  &                         &                          &                      &                         &                          &  &                         &                          \\
\multirow{10}{*}{Scenario c)} & \multirow{5}{*}{Bias} & $\beta_{01}$ & -0.054                  & -0.043                   &  & -0.134                  & -0.083                   &                      & -0.043                  & -0.029                   &  & -0.224                  & -0.081                   \\
                              &                       & $\beta_{11}$ & -0.003                  & -0.005                   &  & 0.015                   & -0.005                   &                      & 0.002                   & 0.000                    &  & 0.068                   & 0.011                    \\
                              &                       & $\beta_{02}$ & -0.029                  & -0.013                   &  & -0.027                  & -0.014                   &                      & -0.038                  & -0.013                   &  & -0.030                  & -0.013                   \\
                              &                       & $\beta_{12}$ & 0.041                   & 0.024                    &  & 0.031                   & 0.017                    &                      & 0.058                   & 0.027                    &  & 0.033                   & 0.018                    \\
                              &                       & $\pi_1$      & -0.001                  & -0.006                   &  & 0.009                   & -0.001                   &                      & -0.002                  & -0.007                   &  & 0.021                   & 0.003                    \\
                              & \multirow{5}{*}{MSE}  & $\beta_{01}$ & 0.023                   & 0.019                    &  & 0.060                   & 0.035                    &                      & 0.011                   & 0.006                    &  & 0.103                   & 0.017                    \\
                              &                       & $\beta_{11}$ & 0.007                   & 0.006                    &  & 0.010                   & 0.007                    &                      & 0.002                   & 0.002                    &  & 0.015                   & 0.003                    \\
                              &                       & $\beta_{02}$ & 0.006                   & 0.006                    &  & 0.006                   & 0.006                    &                      & 0.004                   & 0.003                    &  & 0.004                   & 0.003                    \\
                              &                       & $\beta_{12}$ & 0.004                   & 0.003                    &  & 0.004                   & 0.003                    &                      & 0.005                   & 0.002                    &  & 0.004                   & 0.002                    \\
                              &                       & $\pi_1$      & 0.000                   & 0.000                    &  & 0.001                   & 0.000                    &                      & 0.000                   & 0.000                    &  & 0.001                   & 0.000                    \\
                              &                       &              &                         &                          &  &                         &                          &                      &                         &                          &  &                         &                          \\
\multirow{10}{*}{Scenario d)} & \multirow{5}{*}{Bias} & $\beta_{01}$ & -0.017                  & -0.020                   &  & -0.022                  & -0.021                   &                      & -0.002                  & -0.011                   &  & -0.007                  & -0.017                   \\
                              &                       & $\beta_{11}$ & -0.005                  & -0.009                   &  & -0.004                  & -0.009                   &                      & 0.008                   & 0.004                    &  & 0.006                   & 0.004                    \\
                              &                       & $\beta_{02}$ & 0.007                   & 0.004                    &  & 0.006                   & 0.004                    &                      & 0.001                   & 0.000                    &  & 0.001                   & 0.000                    \\
                              &                       & $\beta_{12}$ & -0.004                  & -0.001                   &  & -0.002                  & 0.000                    &                      & -0.001                  & 0.000                    &  & -0.001                  & 0.000                    \\
                              &                       & $\pi_1$      & 0.002                   & 0.002                    &  & 0.002                   & 0.002                    &                      & 0.002                   & 0.002                    &  & 0.002                   & 0.003                    \\
                              & \multirow{5}{*}{MSE}  & $\beta_{01}$ & 0.011                   & 0.011                    &  & 0.012                   & 0.012                    &                      & 0.004                   & 0.004                    &  & 0.005                   & 0.005                    \\
                              &                       & $\beta_{11}$ & 0.009                   & 0.009                    &  & 0.009                   & 0.008                    &                      & 0.003                   & 0.003                    &  & 0.003                   & 0.003                    \\
                              &                       & $\beta_{02}$ & 0.004                   & 0.004                    &  & 0.004                   & 0.005                    &                      & 0.002                   & 0.002                    &  & 0.002                   & 0.002                    \\
                              &                       & $\beta_{12}$ & 0.002                   & 0.002                    &  & 0.002                   & 0.002                    &                      & 0.001                   & 0.001                    &  & 0.001                   & 0.001                    \\
                              &                       & $\pi_1$      & 0.000                   & 0.000                    &  & 0.000                   & 0.000                    &                      & 0.000                   & 0.000                    &  & 0.000                   & 0.000                    \\ \bottomrule
\end{tabular}
\end{table}

\clearpage
\begin{table}[ht]
\centering
\caption{Values of TPRs and FPRs under for scenarios (a)-(c), under assignment independence across 100 replications.}\label{table tpr fpr under assignment independence for a to c}
\begin{tabular}{llrrrrrrr}
\toprule
                                         &               & \multicolumn{3}{c}{TPR}                                                                             & \multicolumn{1}{l}{} & \multicolumn{3}{c}{FPR}                                                                             \\ \cline{3-5} \cline{7-9} 
\multicolumn{1}{c}{$n$}                                          &               & \multicolumn{1}{c}{(a)} & \multicolumn{1}{c}{(b)} & \multicolumn{1}{c}{(c)} & \multicolumn{1}{l}{} & \multicolumn{1}{c}{(a)} & \multicolumn{1}{c}{(b)} & \multicolumn{1}{c}{(c)} \\ \midrule
\multicolumn{1}{c}{\multirow{3}{*}{250}} & cSALMRM: Outliers      & 1                           &                                 &                                 &                      & 0.028                           & 0.085                           & 0.056                           \\
\multicolumn{1}{c}{}                     & cSALCWM: Outliers & 0.970                                &                            &                                 &                      & 0.023                           & 0.109                           & 0.055                           \\
\multicolumn{1}{c}{}                     & cSALCWM: Good leverage &                                 & 0.970                           &                                 &                      & 0.083                           & 0.036                           & 0.030                           \\

\multicolumn{1}{c}{}                     & cSALCWM: Bad leverage  &                                 &                                 & 0.940                           &                      & 0.001                           & 0.001                           & 0.000                           \\
                                         &               &                                 &                                 &                                 &                      &                                 &                                 &                                 \\
\multirow{3}{*}{500}                     & cSALMRM: Outliers      & 1                           &                                 &                                 &                      & 0.007                           & 0.023                           & 0.028                           \\
                                         & cSALCWM: Outliers &0.970                                 &                            &                                 &                      & 0.007                           & 0.032                           & 0.023                           \\
                                         &cSALCWM: Good leverage &                                 & 1                           &                                 &                      & 0.039                           & 0.012                           & 0.016                           \\
                                         &cSALCWM: Bad leverage  &                                 &                                 & 0.980                           &                      & 0.001                           & 0.002                           & 0.001                           \\ \bottomrule
\end{tabular}
\end{table}

\begin{table}[ht]
\centering
\caption{Values of TPRs and FPRs for scenario (d) under assignment independence across 100 replications.}\label{table tpr fpr under assignment independence for d}
\begin{tabular}{cllrr}
\toprule
\multicolumn{1}{l}{$n$} & & & \multicolumn{1}{c}{TPR} & \multicolumn{1}{c}{FPR} \\ \midrule
\multicolumn{1}{c}{\multirow{2}{*}{250}} & cSALMRM                &  & 0.285                   & 0.056                  \\
& cSALCWM &  & 0.655                   & 0.137                  \\
&  &  &                    &                   \\
\multicolumn{1}{c}{\multirow{2}{*}{500}} & cSALMRM                &  & 0.254                   & 0.031                  \\
& cSALCWM &  & 0.576                   & 0.055                  \\
 \bottomrule
\end{tabular}
\end{table}

\subsection{Comparison under assignment dependence}
The second simulation study aims to compare the models under assignment dependence. In this case, all four data-generating processes share the following parameters
\begin{align*}
    \mu_{X|1}=-3, \quad \mu_{X|2}=3,
    \quad 
    \Sigma_{X|1}=\Sigma_{X|2}=1, \quad
    \alpha_{X|1}=\alpha_{X|2}=0.2.
\end{align*}

Figure~\ref{Fig assignment dep} shows a single replication of the setup, where the black crosses indicate the original SALCWM data points that were replaced. The figure also displays the distinct atypical observations introduced in each scenario under assignment dependence.
\begin{figure}[ht]
    \centering
    \includegraphics[scale=0.6]{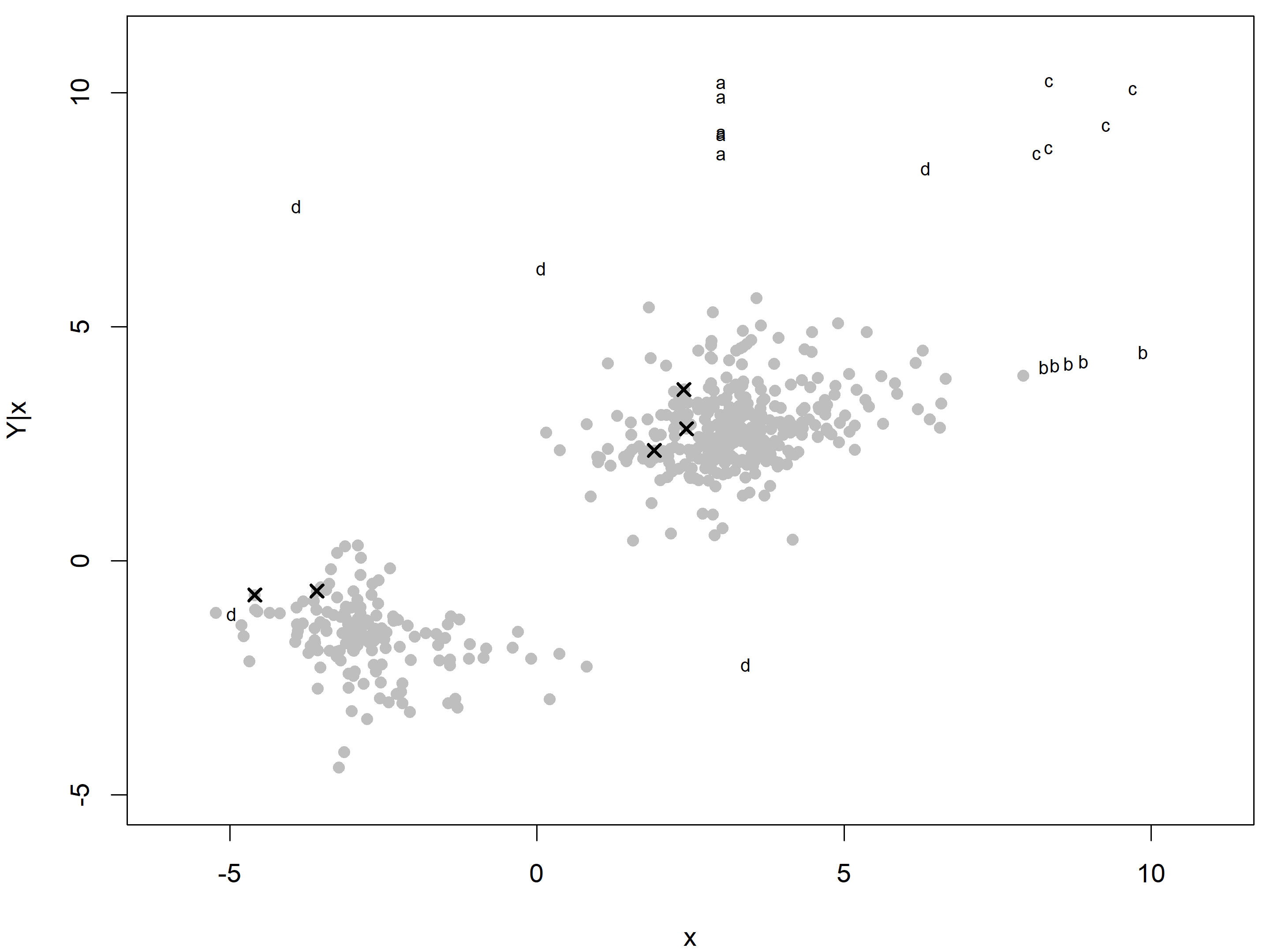}
    \caption{Example replication under assignment dependence. The black crosses indicate original SALCWM data points that have been replaced, while the corresponding replacements are shown for each scenario.}
    \label{Fig assignment dep}
\end{figure}

The bias and MSE values obtained under assignment dependence are reported in \tablename~\ref{table bias mse values under assignment dependence} for sample sizes $n=250$ and $n=500$. Regardless of the scenario, both the bias and MSE values for the CWMs are remarkably lower in absolute value than those of the MRMs. This is expected, as the data's clustering structure strongly depends on two well-separated component distributions for $X$ (reflecting assignment dependence). This underscores the necessity of models capable of accommodating random covariates, a feature naturally handled by CWMs but not by MRMs. This advantage is further corroborated by the TPR and FPR values given in \tablename~\ref{table tpr fpr under assignment dependence for a to c} and \ref{table tpr fpr under assignment dependence for d}, where it is evident that the cSALCWM consistently outperforms the cSALMRM across different scenarios, achieving notably higher TPRs while maintaining very low FPRs.

\begin{table}[ht]
\centering
\caption{Bias and MSE values under assignment dependence across 100 replications.}\label{table bias mse values under assignment dependence}
\begin{tabular}{lllrrrrrrrrrrr}
\toprule
                              &                       &              & \multicolumn{5}{c}{$n=250$}                                                                                                    & \multicolumn{1}{c}{} & \multicolumn{5}{c}{$n=500$}                                                                                                    \\ \cline{4-8} \cline{10-14} 
                              &                       &              & \multicolumn{2}{c}{MRM}                            & \multicolumn{1}{l}{} & \multicolumn{2}{c}{CWM}                            & \multicolumn{1}{c}{} & \multicolumn{2}{c}{MRM}                            & \multicolumn{1}{l}{} & \multicolumn{2}{c}{CWM}                            \\ \cline{4-5} \cline{7-8} \cline{10-11} \cline{13-14} 
                              &                       &              & \multicolumn{1}{c}{SAL} & \multicolumn{1}{c}{cSAL} & \multicolumn{1}{l}{} & \multicolumn{1}{c}{SAL} & \multicolumn{1}{c}{cSAL} & \multicolumn{1}{c}{} & \multicolumn{1}{c}{SAL} & \multicolumn{1}{c}{cSAL} & \multicolumn{1}{l}{} & \multicolumn{1}{c}{SAL} & \multicolumn{1}{c}{cSAL} \\ \midrule
\multirow{10}{*}{Scenario a)} & \multirow{5}{*}{Bias} & $\beta_{01}$ & 0.598                   & 0.596                    &                      & -0.024                  & -0.022                   &                      & 0.571                   & 0.569                    &                      & -0.009                  & -0.007                   \\
                              &                       & $\beta_{11}$ & 0.078                   & 0.076                    &                      & -0.003                  & -0.003                   &                      & 0.074                   & 0.074                    &                      & 0.002                   & 0.002                    \\
                              &                       & $\beta_{02}$ & -1.251                  & -1.231                   &                      & -0.038                  & -0.025                   &                      & -1.441                  & -1.406                   &                      & -0.023                  & -0.002                   \\
                              &                       & $\beta_{12}$ & 0.402                   & 0.401                    &                      & 0.001                   & 0.000                    &                      & 0.453                   & 0.450                    &                      & -0.005                  & -0.005                   \\
                              &                       & $\pi_1$      & -0.227                  & -0.228                   &                      & -0.003                  & -0.003                   &                      & -0.259                  & -0.259                   &                      & -0.003                  & -0.003                   \\
                              & \multirow{5}{*}{MSE}  & $\beta_{01}$ & 2.296                   & 2.296                    &                      & 0.049                   & 0.048                    &                      & 2.227                   & 2.209                    &                      & 0.021                   & 0.021                    \\
                              &                       & $\beta_{11}$ & 0.033                   & 0.033                    &                      & 0.005                   & 0.005                    &                      & 0.025                   & 0.026                    &                      & 0.002                   & 0.002                    \\
                              &                       & $\beta_{02}$ & 1.800                   & 1.765                    &                      & 0.028                   & 0.029                    &                      & 2.128                   & 2.031                    &                      & 0.011                   & 0.011                    \\
                              &                       & $\beta_{12}$ & 0.185                   & 0.185                    &                      & 0.002                   & 0.002                    &                      & 0.210                   & 0.207                    &                      & 0.001                   & 0.001                    \\
                              &                       & $\pi_1$      & 0.060                   & 0.060                    &                      & 0.000                   & 0.000                    &                      & 0.070                   & 0.070                    &                      & 0.000                   & 0.000                    \\
                              &                       &              &                         &                          &                      &                         &                          &                      &                         &                          &                      &                         &                          \\
\multirow{10}{*}{Scenario b)} & \multirow{5}{*}{Bias} & $\beta_{01}$ & 0.216                   & 0.215                    &                      & -0.024                  & -0.025                   &                      & 0.063                   & 0.063                    &                      & -0.010                  & -0.008                   \\
                              &                       & $\beta_{11}$ & 0.039                   & 0.038                    &                      & -0.001                  & -0.002                   &                      & 0.024                   & 0.023                    &                      & 0.003                   & 0.002                    \\
                              &                       & $\beta_{02}$ & -0.777                  & -0.778                   &                      & 0.019                   & 0.015                    &                      & -0.849                  & -0.853                   &                      & 0.020                   & 0.018                    \\
                              &                       & $\beta_{12}$ & 0.268                   & 0.268                    &                      & -0.004                  & -0.004                   &                      & 0.289                   & 0.290                    &                      & -0.003                  & -0.003                   \\
                              &                       & $\pi_1$      & -0.162                  & -0.163                   &                      & -0.003                  & -0.003                   &                      & -0.189                  & -0.190                   &                      & -0.003                  & -0.003                   \\
                              & \multirow{5}{*}{MSE}  & $\beta_{01}$ & 0.713                   & 0.715                    &                      & 0.049                   & 0.049                    &                      & 0.039                   & 0.039                    &                      & 0.020                   & 0.021                    \\
                              &                       & $\beta_{11}$ & 0.021                   & 0.021                    &                      & 0.005                   & 0.005                    &                      & 0.005                   & 0.005                    &                      & 0.002                   & 0.002                    \\
                              &                       & $\beta_{02}$ & 1.036                   & 1.036                    &                      & 0.014                   & 0.014                    &                      & 1.119                   & 1.127                    &                      & 0.004                   & 0.004                    \\
                              &                       & $\beta_{12}$ & 0.119                   & 0.119                    &                      & 0.000                   & 0.000                    &                      & 0.127                   & 0.128                    &                      & 0.000                   & 0.000                    \\
                              &                       & $\pi_1$      & 0.042                   & 0.043                    &                      & 0.000                   & 0.000                    &                      & 0.052                   & 0.052                    &                      & 0.000                   & 0.000                    \\
                              &                       &              &                         &                          &                      &                         &                          &                      &                         &                          &                      &                         &                          \\
\multirow{10}{*}{Scenario c)} & \multirow{5}{*}{Bias} & $\beta_{01}$ & 0.594                   & 0.591                    &                      & -0.009                  & -0.010                   &                      & 0.434                   & 0.435                    &                      & 0.038                   & 0.018                    \\
                              &                       & $\beta_{11}$ & 0.076                   & 0.075                    &                      & 0.004                   & 0.003                    &                      & 0.056                   & 0.056                    &                      & 0.019                   & 0.010                    \\
                              &                       & $\beta_{02}$ & -1.333                  & -1.329                   &                      & -0.115                  & -0.070                   &                      & -1.397                  & -1.398                   &                      & -0.126                  & -0.056                   \\
                              &                       & $\beta_{12}$ & 0.449                   & 0.448                    &                      & 0.028                   & 0.016                    &                      & 0.462                   & 0.462                    &                      & 0.031                   & 0.012                    \\
                              &                       & $\pi_1$      & -0.248                  & -0.249                   &                      & -0.003                  & -0.003                   &                      & -0.271                  & -0.271                   &                      & -0.002                  & -0.003                   \\
                              & \multirow{5}{*}{MSE}  & $\beta_{01}$ & 2.333                   & 2.333                    &                      & 0.096                   & 0.095                    &                      & 1.798                   & 1.798                    &                      & 0.139                   & 0.083                    \\
                              &                       & $\beta_{11}$ & 0.033                   & 0.033                    &                      & 0.010                   & 0.009                    &                      & 0.015                   & 0.015                    &                      & 0.017                   & 0.009                    \\
                              &                       & $\beta_{02}$ & 1.869                   & 1.870                    &                      & 0.043                   & 0.041                    &                      & 1.978                   & 1.981                    &                      & 0.027                   & 0.015                    \\
                              &                       & $\beta_{12}$ & 0.210                   & 0.211                    &                      & 0.003                   & 0.003                    &                      & 0.216                   & 0.216                    &                      & 0.002                   & 0.001                    \\
                              &                       & $\pi_1$      & 0.066                   & 0.066                    &                      & 0.000                   & 0.000                    &                      & 0.075                   & 0.075                    &                      & 0.000                   & 0.000                    \\
                              &                       &              &                         &                          &                      &                         &                          &                      &                         &                          &                      &                         &                          \\
\multirow{10}{*}{Scenario d)} & \multirow{5}{*}{Bias} & $\beta_{01}$ & 0.231                   & 0.231                    &                      & -0.069                  & -0.049                   &                      & 0.040                   & 0.041                    &                      & -0.053                  & -0.023                   \\
                              &                       & $\beta_{11}$ & 0.029                   & 0.028                    &                      & -0.016                  & -0.013                   &                      & 0.000                   & 0.000                    &                      & -0.012                  & -0.005                   \\
                              &                       & $\beta_{02}$ & -1.054                  & -1.069                   &                      & 0.025                   & 0.019                    &                      & -1.285                  & -1.311                   &                      & 0.032                   & 0.027                    \\
                              &                       & $\beta_{12}$ & 0.354                   & 0.360                    &                      & -0.004                  & -0.003                   &                      & 0.425                   & 0.434                    &                      & -0.008                  & -0.007                   \\
                              &                       & $\pi_1$      & -0.209                  & -0.212                   &                      & 0.002                   & 0.001                    &                      & -0.261                  & -0.266                   &                      & 0.003                   & 0.003                    \\
                              & \multirow{5}{*}{MSE}  & $\beta_{01}$ & 0.877                   & 0.879                    &                      & 0.064                   & 0.060                    &                      & 0.097                   & 0.097                    &                      & 0.038                   & 0.031                    \\
                              &                       & $\beta_{11}$ & 0.020                   & 0.020                    &                      & 0.006                   & 0.006                    &                      & 0.016                   & 0.016                    &                      & 0.003                   & 0.002                    \\
                              &                       & $\beta_{02}$ & 1.456                   & 1.479                    &                      & 0.030                   & 0.029                    &                      & 1.764                   & 1.804                    &                      & 0.012                   & 0.011                    \\
                              &                       & $\beta_{12}$ & 0.161                   & 0.164                    &                      & 0.002                   & 0.002                    &                      & 0.192                   & 0.196                    &                      & 0.001                   & 0.001                    \\
                              &                       & $\pi_1$      & 0.055                   & 0.056                    &                      & 0.000                   & 0.000                    &                      & 0.072                   & 0.074                    &                      & 0.000                   & 0.000                    \\ \bottomrule
\end{tabular}
\end{table}

\begin{table}[ht]
\centering
\caption{Values of TPRs and FPRs under for scenarios (a)-(c), under assignment dependence across 100 replications.}\label{table tpr fpr under assignment dependence for a to c}
\begin{tabular}{llrrrrrrr}
\toprule
                                         &               & \multicolumn{3}{c}{TPR}                                                                             & \multicolumn{1}{l}{} & \multicolumn{3}{c}{FPR}                                                                             \\ \cline{3-5} \cline{7-9} 
\multicolumn{1}{c}{$n$}                                          &               & \multicolumn{1}{l}{(a)} & \multicolumn{1}{l}{(b)} & \multicolumn{1}{l}{(c)} & \multicolumn{1}{l}{} & \multicolumn{1}{l}{(a)} & \multicolumn{1}{l}{(b)} & \multicolumn{1}{l}{(c)} \\ \midrule
\multicolumn{1}{c}{\multirow{3}{*}{250}} & cSALMRM: Outliers      & 0.875                           &                                 &                                 &                      & 0.005                           & 0.020                           & 0.009                           \\
\multicolumn{1}{c}{}                     & cSALCWM: Outliers & 0.960                                &                            &                                 &                      & 0.028                           & 0.036                           & 0.048                           \\
\multicolumn{1}{c}{}                     & cSALCWM: Good leverage &                                 & 0.580                           &                                 &                      & 0.045                           & 0.055                           & 0.067                           \\

\multicolumn{1}{c}{}                     & cSALCWM: Bad leverage  &                                 &                                 & 0.750                           &                      & 0.001                           & 0.007                           & 0.005                           \\
                                         &               &                                 &                                 &                                 &                      &                                 &                                 &                                 \\
\multirow{3}{*}{500}                     & cSALMRM: Outliers      & 0.970                           &                                 &                                 &                      & 0.002                           & 0.055                           & 0.000                           \\
                                         & cSALCWM: Outliers &0.980                                 &                            &                                 &                      & 0.025                           & 0.075                           & 0.035                           \\
                                         &cSALCWM: Good leverage &                                 & 0.770                           &                                 &                      & 0.057                           & 0.041                           & 0.044                           \\
                                         &cSALCWM: Bad leverage  &                                 &                                 & 0.954                           &                      & 0.001                           & 0.006                           & 0.001                           \\ \bottomrule
\end{tabular}
\end{table}

\begin{table}[ht]
\centering
\caption{Values of TPRs and FPRs for scenario (d) under assignment dependence across 100 replications.}\label{table tpr fpr under assignment dependence for d}
\begin{tabular}{cllrr}
\toprule
\multicolumn{1}{l}{$n$} & & & \multicolumn{1}{c}{TPR} & \multicolumn{1}{c}{FPR} \\ \midrule
\multicolumn{1}{c}{\multirow{2}{*}{250}} & cSALMRM                &  & 0.185                   & 0.014                  \\
& cSALCWM &  & 0.515                   & 0.099                  \\
&  &  &                    &                   \\
\multicolumn{1}{c}{\multirow{2}{*}{500}} & cSALMRM                &  & 0.154                   & 0.004                  \\
& cSALCWM &  & 0.548                   & 0.093                  \\
 \bottomrule
\end{tabular}
\end{table}

\section{Application} \label{Section Application}
In this application, we conduct a sensitivity analysis based on the Australian Institute of Sport (AIS) dataset.
The AIS data, available in the \texttt{sn} package \citep{azzalini2020r} in \texttt{R}, contains measurements on 102 male and 100 female athletes ($n=202$, $G=2$), and has previously been studied in the context of mixture models \citep{otto2024refreshing,soffritti2011multivariate,li2016robust, gallaugher2022multivariate}. We focus on a subset of six variables, which include red cell count (RCC), white cell count (WCC), body mass index (BMI), sum of skin folds (SSF), body fat percentage (BFT), and lean body mass (LBM). The blood-related variables (RCC and WCC) are treated as responses ($\bY$), while the biometric variables (BMI, SSF, BFT, and LBM) serve as covariates ($\bX$). The pairwise scatter plots of the data, coloured according to sex, with sex being considered as the variable defining the group structure, are given in \figurename~\ref{fig ais}.

\subsection{Australian Institute of Sport data}
As an initial step, we fit the SALCWM and cSALCWM to the AIS data.
The BIC values of the SALCWM and cSALCWM are presented in \tablename~\ref{Table aic bic} and compared to the Gaussian CWM (GCWM) and CGCWM for different numbers of components ($G=1,2,3$). We note that the models with $G=2$ components consistently yield the lowest BIC values across all three considered CWMs. Among them, the SALCWM achieves the lowest overall BIC, indicating the best model fit according to this criterion. As discussed in Section \ref{Sec comparison with mixtures of laplace regression}, direct comparison between CWMs and MRMs are not appropriate. Consequently, BIC values for the SALMRM, cSALMRM, GMRM, and CGMRM are not included in the comparison.
However, these models are included when evaluating the clustering performance using the ARI, as shown in \tablename~\ref{Table ais ari}.  In this context,  all the MRMs exhibit very low ARI values, reflecting poor clustering performance on the \texttt{AIS} dataset.
The performance of the CWMs -- particularly the SALCWM and cSALCWM -- is notably superior. Both the SALCWM and cSALCWM achieve the highest ARI value of 0.961, indicating strong agreement with the actual class labels (sex). However, the cSALCWM does not detect any outliers, which is consistent with the inflation parameter estimates of the cSALCWM $\hat\boldeta_{\bX}$ and $\hat\boldeta_{\bY}$ both being very close to $\boldsymbol{1}$.


\begin{figure}[ht]
    \centering
    \includegraphics[width=1\linewidth]{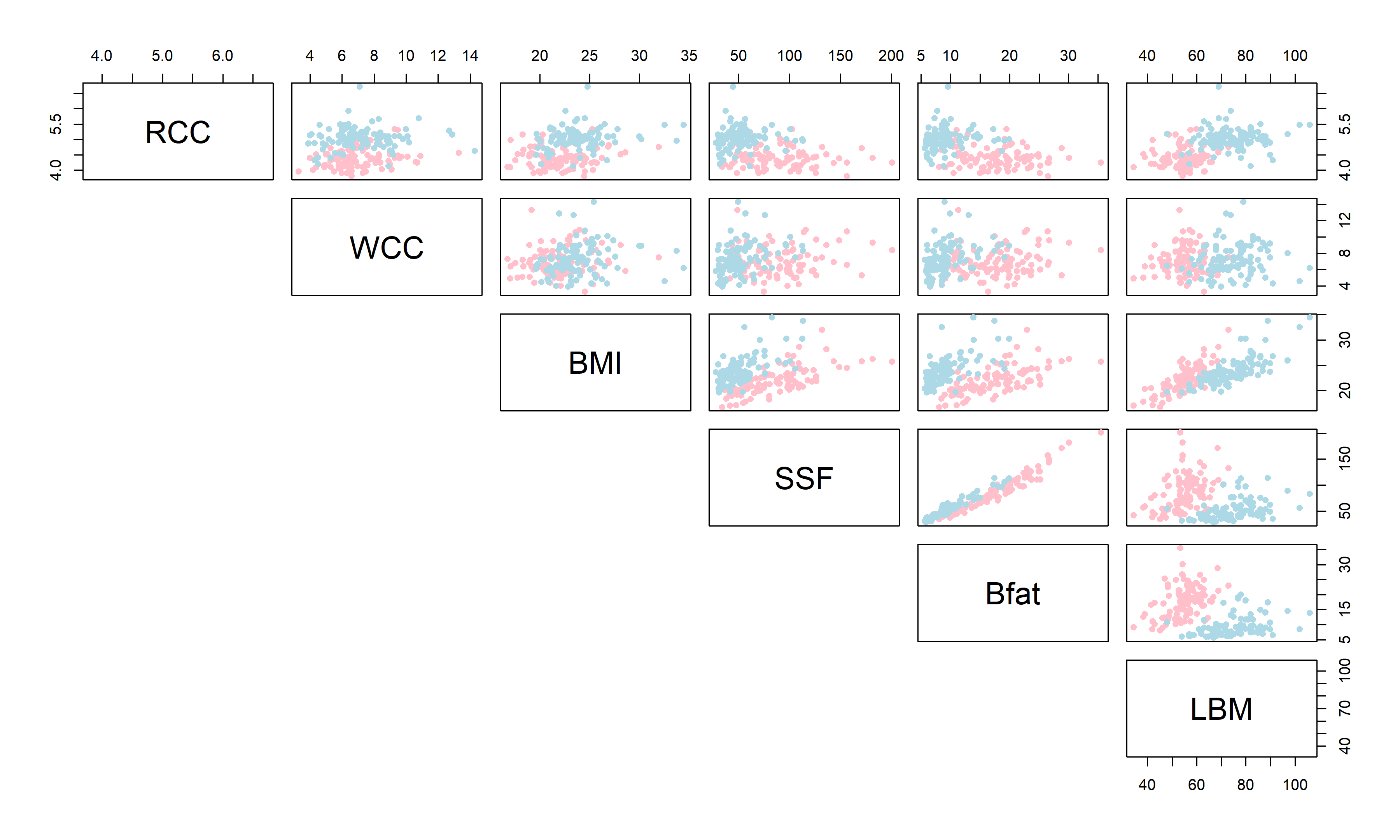}
    \caption{Pairwise scatter plots of the \texttt{AIS} dataset, with points colour-coded by actual sex: blue for males and pink for females.}
    \label{fig ais}
\end{figure}

\begin{table}[ht]
\centering
\caption{BIC values of CWMs on the \texttt{AIS} dataset for different number of components.}
\label{Table aic bic}
\begin{tabular}{lrrr}
\toprule
        & \multicolumn{1}{c}{$G=1$} & \multicolumn{1}{c}{$G=2$} & \multicolumn{1}{c}{$G=3$} \\ \midrule
GCWM   & 6074.424                  & 5975.203                 & 5989.772                  \\
CGCWM   & 6095.624                  & 5986.094                  & 6022.245                  \\
SALCWM  & 6006.121                  & \textbf{5907.134}                  & 6011.640                  \\
cSALCWM & 6027.354                  & 5949.600                 & 6075.339                  \\ \bottomrule
\end{tabular}
\end{table}
\begin{table}[ht]
\centering
\caption{Comparison of the ARI performance of the best fitting models according to the BIC on the \texttt{AIS} dataset. }
\label{Table ais ari}
\begin{tabular}{lr}
\toprule
Model   & \multicolumn{1}{c}{ARI} \\ \midrule
CGMRM   & 0.000                   \\
CGMRM   & 0.000                   \\
SALMRM  & 0.010                   \\
cSALRM  & 0.010                   \\
GCWM   & 0.657                   \\
CGCWM   & 0.707                   \\
\textbf{SALCWM}  & \textbf{0.961}                   \\
\textbf{cSALCWM} & \textbf{0.961}                   \\ \bottomrule
\end{tabular}
\end{table}

\subsection{Australian Institute of Sport data with noise}
To illustrate the effect of atypical value on model performance using real data, we conduct a sensitivity analysis by adding noise and assessing its impact on both model fitting and clustering performance. To this end, we modify the original \texttt{AIS} data by including 10 noisy points from a uniform distribution over the hyper-rectangle. The hyper-rectangle  $(3,10)\times(2,20)\times(15,40)\times(25,230)\times(5,45)\times(30,120)$ is chosen to contain all the data. While alternative strategies could guide the selection of these intervals, simply choosing intervals that exceed the empirical ranges given in \tablename~\ref{table range AIS} suffices to illustrate the impact of atypical values. The pairwise scatter plots of the data, coloured according to the true data classification, and illustrating the contaminated points, are given in \figurename~\ref{fig ais with noise}. 

\begin{table}[ht]
\caption{Range of the variables in the \texttt{AIS} dataset.}\label{table range AIS}
\centering
\begin{tabular}{lrrrrrr}
\toprule
    & \multicolumn{1}{c}{RCC} & \multicolumn{1}{c}{WCC} & \multicolumn{1}{c}{BMI} & \multicolumn{1}{c}{SSF} & \multicolumn{1}{c}{Bfat} & \multicolumn{1}{c}{LBM} \\ \midrule
min & 3.800                   & 3.300                   & 16.750                  & 28.000                  & 5.630                    & 34.360                  \\
max & 6.720                   & 14.300                  & 34.420                  & 200.800                 & 35.520                   & 106.000                 \\ \bottomrule
\end{tabular}
\end{table}

\begin{figure}[ht]
    \centering
\includegraphics[width=1\linewidth]{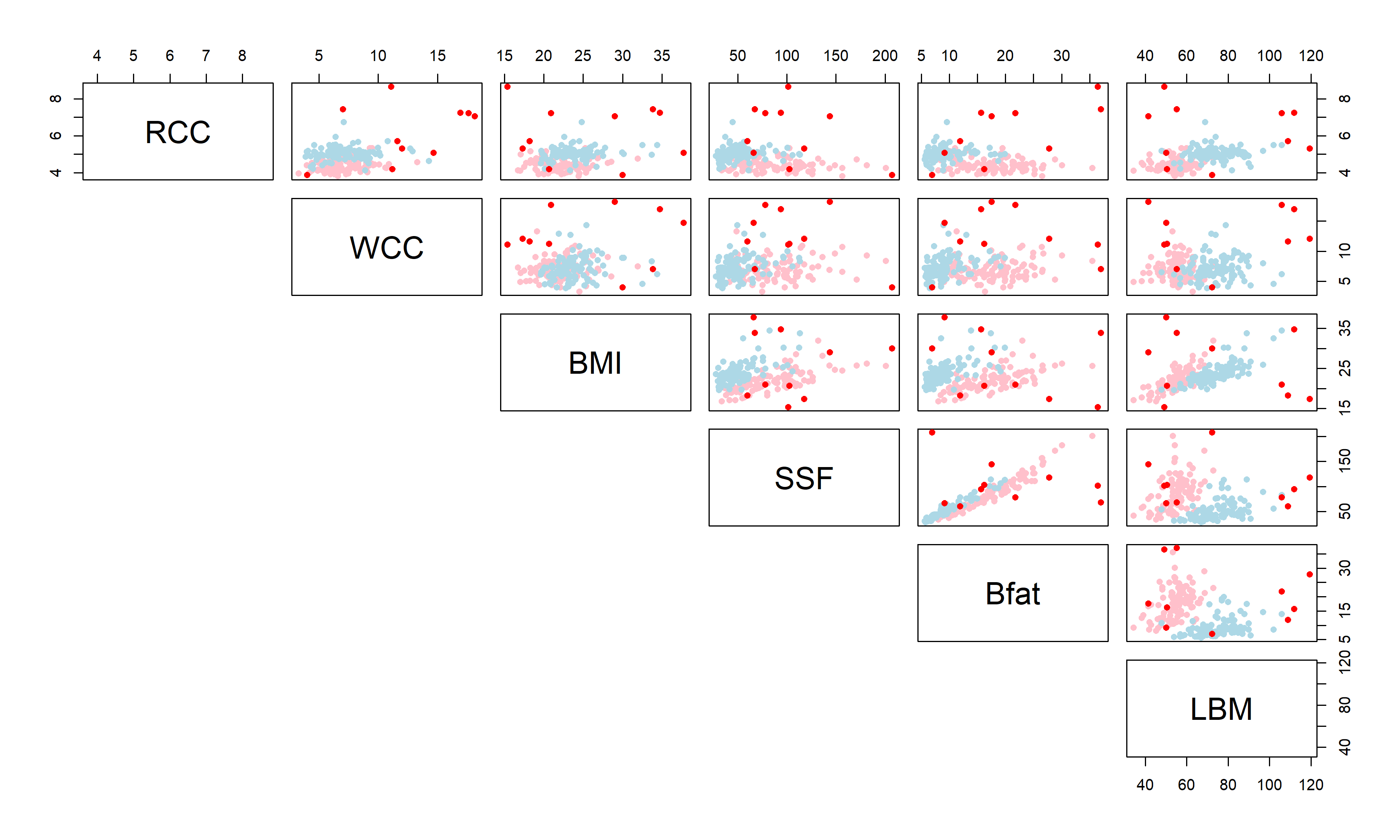}
    \caption{Pairwise scatter plots of the \texttt{AIS} dataset with 20 uniform noise points indicated in red.}
    \label{fig ais with noise}
\end{figure}

\tablename~\ref{Table aic bic contaminated AIS} presents the BIC values, in correspondence of $G\in \{1,2,3\}$, for the CGCWM, SALCWM and cSALCWM for the contaminated \texttt{AIS} dataset. Again, both the SALCWM and cSALCWM outperformed the CGCWM, regardless of the number of components considered. However, in this case, the best fitting model is the cSALCWM with $G=2$, with contamination parameter estimates: proportion of leverage $\hat\bdelta_{\bX}=(0.083,0.010),$ degree of leverage $\hat\boldeta_{\bX}=(43.084,1.000),$ proportion of outliers $\hat\bdelta_{\bY}=(0.050,0.045),$  and degree of outlierness $\hat\boldeta_{\bY}=(29.450,15.568)$. 9 of the 10 added noise points were detected as atypical values (see \tablename~\ref{table rule of points in a regression analysis}).
To assess the clustering performance of the models in the presence of atypical values, we compute the ARI using only the original AIS observations. Again, both the SALCWM and cSALCWM outperform the classical model; however, in this case, the cSALCWM achieves the highest ARI, confirming its enhanced ability to handle atypical values effectively.

\begin{table}[ht]
\centering
\caption{BIC values of CWMs on the \texttt{AIS} dataset with noise for different number of components.}
\label{Table aic bic contaminated AIS}
\begin{tabular}{lrrr}
\toprule
        & \multicolumn{1}{c}{$G=1$} & \multicolumn{1}{c}{$G=2$} & \multicolumn{1}{c}{$G=3$} \\ \midrule
GCWM   &  7314.430                 &    7026.724               & 6615.774                  \\
CGCWM   &  7230.856                 &    6610.410               & 6635.030                  \\
SALCWM  &  6789.344                 &     6552.267              &  6525.261                \\
cSALCWM &  6613.455                 &    \textbf{6523.194}               & 6544.719                  \\ \bottomrule
\end{tabular}
\end{table}

\begin{table}[ht]
\centering
\caption{Comparison of the ARI performance of the best fitting models according to the BIC on the \texttt{AIS} dataset with noise. }
\label{Table ais ari contaminated AIS}
\begin{tabular}{lr}
\toprule
Model   & \multicolumn{1}{c}{ARI} \\ \midrule
GMRM   & 0.000                   \\
CGMRM   & 0.001                   \\
SALMRM  & 0.005                   \\
cSALRM  & 0.006                   \\
GCWM   & 0.657                   \\
CGCWM   & 0.775                   \\
SALCWM  & 0.866                   \\
\textbf{cSALCWM} & \textbf{0.903}                   \\ \bottomrule
\end{tabular}
\end{table}

\section{Discussion}\label{Section conclusion}
In this paper, we introduced the shifted asymmetric Laplace (SAL) cluster-weighted model (SALCWM) as a mixture of multivariate linear regression models with multiple random covariates that allows for clustering of responses and covariates exhibiting varying degrees of skewness, which can vary across clusters. To address the impact of atypical observations, we further propose the contaminated SALCWM (cSALCWM), which simultaneously mitigates the influence of atypical values and facilitates their detection, distinguishing specifically between typical points, mild outliers, good leverage points, and bad leverage points.

An expectation-conditional maximization algorithm is discussed for maximum likelihood parameter estimation, supported by identifiability conditions. Simulation studies also illustrate the effectiveness of the cSALCWM in both mitigating the impact of atypical points and accurately identifying them. The proposed CWMs are also compared to their corresponding finite mixture of regression models (MRMs), where covariates are treated as fixed. Specifically, it is shown that by ignoring the distribution of the covariates, MRMs may fail to recognize the underlying group structure in the data. A real-world application to sports data further illustrates the potential of the proposed models. Importantly, the models are applicable across a wide range of fields, without being restricted to any specific domain.

Possible extensions of this work could include exploring parsimonious structures for the scale matrices for the SAL and cSAL distribution, as suggested by \citet{tortora2024laplace}, and incorporating them into a cluster-weighted modelling framework, similar to the approach of \cite{dang2017multivariate} for the Gaussian CWM.

\section*{Acknowledgments}
This work was based upon research supported in part by the National Research Foundation (NRF) of South Africa (SA), grant RA201125576565, nr 145681 (Ferreira); NRF ref.~SRUG2204203965, UID 119109 (Bekker); the Department of Research and Innovation at the University of Pretoria (SA), the Faculty of Science at the University of the Witwatersrand, as well as the DSI-NRF Centre of Excellence in Mathematical and Statistical Sciences (CoE-MaSS). The opinions expressed and conclusions arrived at are those of the authors and are not necessarily to be attributed to the NRF. In addition, this work was supported by \textit{NSF grant N0. 2209974} (Tortora). Punzo acknowledges the support by the Italian Ministry of University and Research (MUR) under the PRIN 2023 grant number 2022XRHT8R (CUP: E53D23005950006), as part of ‘The SMILE Project: Statistical Modelling and Inference to Live the Environment’, funded by the European Union – Next Generation EU.

\section*{Data availability statement}
All datasets considered in this paper are freely available on the internet.

\section*{Disclosure statement}
The authors declared no potential conflicts of interest with respect to the research, authorship, and/or publication of this article.

\clearpage
\bibliographystyle{chicago}
\bibliography{database.bib}

\appendix 

\section{Proofs} \label{Sec ProofinApp}

\renewcommand{\thefigure}{A.\arabic{figure}}
\setcounter{figure}{0}

\subsection{Proof of Proposition 1}\label{Section appendix proof proposition 1}
\begin{proof}
Under the assumptions of the Proposition \ref{proposition mix of regression special case of CWM}, \eqref{eq csal mixture of regression dynamic weights} simplifies as
\begin{align*}
    p(\by|\bx;\btheta)=&\frac{1}{{\sum\limits_{j=1}^G}\pi_jf_{\text{cSAL}}\left(\bx;\bmu_{\bX},\bSigma_{\bX},\balpha_{\bX},\delta_{\bX},\eta_{\bX}\right)}\notag\\
    &\times\sum_{g=1}^G\pi_gf_{\text{cSAL}}\left(\bx;\bmu_{\bX},\bSigma_{\bX},\balpha_{\bX},\delta_{\bX},\eta_{\bX}\right)f_{\text{cSAL}}\left(\bx;\bmu\left(\bx;\bbeta_g\right),\bSigma_{\bY|g},\balpha_{\bY|g},\delta_{\bY|g},\eta_{\bY|g}\right)\\
    &=\sum_{g=1}^G\pi_gf_{\text{cSAL}}\left(\bx;\bmu\left(\bx;\bbeta_g\right),\bSigma_{\bY|g},\balpha_{\bY|g},\delta_{\bY|g},\eta_{\bY|g}\right),
\end{align*}
which corresponds to the conditional distribution from a mixture of cSAL regression models as defined in \eqref{eq mixtures of cSAL regression}.
\end{proof}
\subsection{Proof of Proposition 2}\label{Section Appendix identifiability}
\begin{proof}Suppose that 
\begin{align}\label{eq proof identifiability}
    p(\bx,\by;\btheta)=p(\bx,\by;\tilde{\btheta})
\end{align}
where $p(\cdot)$ is defined in \eqref{eq general cwm}. Integrating over $\by$ from \eqref{eq proof identifiability} yields
\begin{align}\label{eq proof identifiability 2}
    \sum_{g=1}^G\pi_g f\left(\bx;\bmu_{\bX|g},\bSigma_{\bX|g},\balpha_{\bX|g},\delta_{\bX|g},\eta_{\bX|g}\right)&= \sum_{j=1}^{\tilde{G}}\tilde{\pi}_g f\left(\bx;\tilde{\bmu}_{\bX|g},\tilde{\bSigma}_{\bX|g},\tilde{\balpha}_{\bX|g},\tilde{\delta}_{\bX|g},\tilde{\eta}_{\bX|g}\right)\notag\\
    p(\bx;\bpi,\btheta_{\bX})&=p(\bx;\tilde{\bpi},\tilde{\btheta}_{\bX})
\end{align}
corresponding to the marginal distribution of $\bX$, $p(\bx;\bmu,\bSigma,\balpha,\delta,\eta)$.
Dividing the left-hand (right-hand) side of \eqref{eq proof identifiability} by the left-hand (right-hand) side of \eqref{eq proof identifiability 2} leads to
\begin{align}\label{eq proof identifiability 3}
    &\sum_{g=1}^{G}\frac{\pi_gf\left(\bx;\bmu_{\bX|g},\bSigma_{\bX|g},\balpha_{\bX|g},\delta_{\bX|g},\eta_{\bX|g}\right)}{\sum_{t=1}^G\pi_gf\left(\bx;\bmu_{\bX|t},\bSigma_{\bX|t},\balpha_{\bX|t},\delta_{\bX|t},\eta_{\bX|t}\right)}f\left(\by;\bmu_{\bY|g},\bSigma_{\bY|g},\balpha_{\bY|g},\delta_{\bY|g},\eta_{\bY|g}\right)\notag\\
    &=\sum_{j=1}^{\tilde{G}}\frac{\tilde{\pi}_jf\left(\bx;\tilde{\bmu}_{\bX|j},\tilde{\bSigma}_{\bX|j},\tilde{\balpha}_{\bX|j},\tilde{\delta}_{\bX|j},\tilde{\eta}_{\bX|j}\right)}{\sum_{t=1}^G\tilde{\pi}_gf\left(\bx;\tilde{\bmu}_{\bX|t},\tilde{\bSigma}_{\bX|t},\tilde{\balpha}_{\bX|t},\tilde{\delta}_{\bX|t},\tilde{\eta}_{\bX|t}\right)}f\left(\by;\tilde{\bmu}_{\bY|j},\tilde{\bSigma}_{\bY|j},\tilde{\balpha}_{\bY|j},\tilde{\delta}_{\bY|j},\tilde{\eta}_{\bY|j}\right)\notag\\
    &p(\by|\bx;\btheta)=p(\by|\bx;\tilde{\btheta}).
\end{align}
For each fixed value of $\bx$, $p(\by|\bx;\btheta)$ and $p(\by|\bx;\tilde{\btheta})$ are mixtures of $d_{\bY}$-variate cSAL distributions for $\bY$.
Now, recall from Section \ref{Section Methodological} that the location parameter $\bmu_{\bY|g}$ of the $d_{\bY}$-variate cSAL distribution of $\bY$ in the $g$-th mixture component is related to the covariates $\bX$, through the regression coefficients $\bbeta_g$, by the relation $\bbeta_g^\top\bx^*$, $g=1,\dots,G.$ Define the set of all covariates points $\bx$ which can be used to distinct different regression coefficients $\bbeta_g$ by different values  of $\bmu_{\bY}(\bx;\bbeta_g)$ as
\begin{align*}
    \chi=\left\{\bx \in \mathcal{R}^{d_{\bX}}\right.  : &\left. \forall g, l\in\left\{1,\dots,g\right\}\text{ and } s,t\in\left\{1,\dots,\tilde{g}\right\},\right.\\
    &\left.\bmu_{\bY}\left(\bx;\bbeta_g\right)=\bmu_{\bY}\left(\bx;\bbeta_l\right) \implies \bbeta_g=\bbeta_l,\right.\\
    &\left.\bmu_{\bY}\left(\bx;\bbeta_g\right)=\bmu_{\bY}\left(\bx;\tilde{\bbeta}_s\right) \implies 
    \bbeta_g=\tilde{\bbeta}_s,\right.\\
 &\left.\bmu_{\bY}\left(\bx;\tilde{\bbeta}_s\right)=\bmu_{\bY}\left(\bx;\tilde{\bbeta}_t\right) \implies \tilde{\bbeta}_t=\tilde{\bbeta}_s
    \right\}.
\end{align*}
Note that $\chi$ is the complement of a finite union of hyperplanes of $\mathcal{R}^{d_{\bX}}$. Therefore,
\begin{align*}
    \int_{\chi}p(\bpi,\bmu_{\bX},\bSigma_{\bX},\balpha_{\bX},\bdelta_{\bX},\boldeta_{\bX})d\bx=1.
\end{align*}
For $\bx\in\chi$, all pairs $\left(\bmu_{\bY}\left(\bx;\bbeta_g\right),\bSigma_{\bY|g},\balpha_{\bY|g},\delta_{\bY|g},\eta_{\bY|g}\right)$, $g=1,\dots,G$, are pairwise distinct because all $\left(\bbeta_g,\bSigma_{\bY|g},\balpha_{\bY|g},\delta_{\bY|g},\eta_{\bY|g}\right)$, $g=1,\dots,G$, are pairwise distinct for the hypothesis of the theorem. Since for each fixed value of $\bx$, model \eqref{eq proof identifiability 3} is a mixture of $d_{\bY}$-variate cSAL distributions, which being identifiable \citep{tortora2024laplace} implies that $G=\tilde{G}$ and that, for each $g \in \left\{1, \dots,G\right\}$, there exists a $j \in \left\{1, \dots,G\right\}$ such that
\begin{align*}
    \bbeta_g=\tilde{\bbeta}_j, \quad \bSigma_{\bY|g}=\tilde{\bSigma}_{\bY|j}, \quad \balpha_{\bY|g}=\tilde{\balpha}_{\bY|j}, \quad \delta_{\bY|g}=\tilde\delta_{\bY|j}, \quad \eta_{\bY|g}=\tilde\eta_{\bY|j}
\end{align*}
and
\begin{align}\label{eq proof identifiability 4}
\frac{\pi_gf\left(\bx;\bmu_{\bX|g},\bSigma_{\bX|g},\balpha_{\bX|g},\delta_{\bX|g},\eta_{\bX|g}\right)}{\sum\limits_{t=1}^G\pi_gf\left(\bx;\bmu_{\bX|t},\bSigma_{\bX|t},\balpha_{\bX|t},\delta_{\bX|t},\eta_{\bX|t}\right)}=\frac{\tilde{\pi}_jf\left(\bx;\tilde{\bmu}_{\bX|j},\tilde{\bSigma}_{\bX|j},\tilde{\balpha}_{\bX|j},\tilde{\delta}_{\bX|j},\tilde{\eta}_{\bX|j}\right)}{\sum\limits_{t=1}^G\tilde{\pi}_gf\left(\bx;\tilde{\bmu}_{\bX|t},\tilde{\bSigma}_{\bX|t},\tilde{\balpha}_{\bX|t},\tilde{\delta}_{\bX|t},\tilde{\eta}_{\bX|t}\right)}.
\end{align}
Now, based on \eqref{eq proof identifiability 3}, the equality in \eqref{eq proof identifiability 4} simplifies to
\begin{align}\label{eq proof identifiability 5}
\pi_gf\left(\bx;\bmu_{\bX|g},\bSigma_{\bX|g},\balpha_{\bX|g},\delta_{\bX|g},\eta_{\bX|g}\right)=\tilde{\pi}_jf\left(\bx;\tilde{\bmu}_{\bX|j},\tilde{\bSigma}_{\bX|j},\tilde{\balpha}_{\bX|j},\tilde{\delta}_{\bX|j},\tilde{\eta}_{\bX|j}\right),\quad \forall\bx\in\chi.
\end{align}
Integrating \eqref{eq proof identifiability 5} over $\bx\in\chi$ yields $\pi_g=\tilde\pi_j$. Therefore, condition \eqref{eq proof identifiability 5} further simplifies to 
\begin{align*}
f\left(\bx;\bmu_{\bX|g},\bSigma_{\bX|g},\balpha_{\bX|g},\delta_{\bX|g},\eta_{\bX|g}\right)=f\left(\bx;\tilde{\bmu}_{\bX|j},\tilde{\bSigma}_{\bX|j},\tilde{\balpha}_{\bX|j},\tilde{\delta}_{\bX|j},\tilde{\eta}_{\bX|j}\right),\quad \forall\bx\in\chi.
\end{align*}
The equalities $\bmu_{\bX|g}=\tilde{\bmu}_{\bX|j},$ $ \bSigma_{\bX|g}=\tilde{\bSigma}_{\bX|j}, $ $ \balpha_{\bX|g}=\tilde{\balpha}_{\bX|j},$  $ \delta_{\bX|g}=\tilde\delta_{\bX|j}, $ $ \eta_{\bX|g}=\tilde\eta_{\bX|j}$ simply arise of the identifiability of cSAL distribution, and this completes the proof.
\end{proof}

\subsection{Updates in the first CM step}

\begin{proof}
The updates for $\bmu_{\bX}$ and $\balpha_{\bX}$ can be obtained through the first partial derivatives
    \begin{align}
        &\frac{\partial Q_3\left(\bmu_{\bX},\bSigma_{\bX},\balpha_{\bX},\boldeta_{\bX}|\btheta^{(r)}\right)}{\partial\bmu_{\bX|g}}\notag\\
        &=-\frac{1}{2}\sum^{n}_{i=1}z_{ig}^{(r)}\left(\left(1-v_{ig}^{(r)}\right)\mathrm{E}_{2_{\bX|ig}}+\frac{v_{ig}^{(r)}}{\eta_{\bX|g}^{(r)}}\mathrm{\tilde{E}}_{2_{\bX|ig}}\right)\left(-2\bSigma_{\bX|g}^{-1}\left(\bx_i-\bmu_{\bX|g}\right)\right)\notag\\
        &-\sum^{n}_{i=1}z_{ig}^{(r)}\left(1-v_{ig}^{(r)}+\frac{v_{ig}^{(r)}}{\sqrt{\eta_{\bX|g}^{(r)}}}\right)\left(-\bSigma_{\bX|g}^{-1}\balpha_{\bX|g}\right)\label{eq Q3partialmu},
    \end{align}

        \begin{align}
        &\frac{\partial Q_3\left(\bmu_{\bX},\bSigma_{\bX},\balpha_{\bX},\boldeta_{\bX}|\btheta^{(r)}\right)}{\partial\balpha_{\bX|g}}\notag\\
        &=\sum^{n}_{i=1}z_{ig}\left(1-v_{ig}^{(r)}+\frac{v_{ig}^{(r)}}{\sqrt{\eta_{\bX|g}^{(r)}}}\right)\left(\bSigma_{\bX|g}^{-1}\left(\bx_i-\bmu_{\bX|g}\right)\right)\notag\\
        &-\sum^n_{i=1}z_{ig}\left(\left(1-v_{ig}^{(r)}\right)\mathrm{E}_{1_{\bX|ig}}+v_{ig}^{(r)}\mathrm{\tilde{E}}_{1_{\bX|ig}}\right)\bSigma_{\bX|g}^{-1}\balpha_{\bX|g},\quad g=1,\dots,k.\label{eq Q3partialalpha}
    \end{align}
    \text{By equating \eqref{eq Q3partialmu} and \eqref{eq Q3partialalpha} to the null vector and simultaneously solving for $\bmu_{\bX}$ and $\balpha_{\bX}$, we obtain the updates. }
\end{proof}

\begin{proof}
Similarly, updates $\balpha_{\bY}$ can be obtained through the first partial derivatives   

    \begin{align}
        &\frac{\partial Q_{5}\left(\bbeta,\bSigma_{\bY},\balpha_{\bY},\boldeta_{\bY}|\btheta^{(r)}\right)}{\partial\balpha_{\bY|g}}\notag\\
        &=\frac{\partial}{\partial\balpha_{\bY|g}}\left[\sum_{i=1}^n\sum_{g=1}^Gz_{ig}^{(r)}\left(\left(1-u_{ig}^{(r)}\right)+\frac{u_{ig}^{(r)}}{\sqrt{\eta_{\bY|g}}}\right) \left(\by_i-\bbeta_g^\top \bx_{i}^*\right)^\top\bSigma^{-1}_{\bY|g}\balpha_{\bY|g}\right.\notag\\
        &\left. -\frac{1}{2}\sum_{i=1}^n\sum_{g=1}^Gz_{ig}^{(r)}\left(\left(1-u_{ig}^{(r)}\right)\mathrm{E}_{1_{\bY|ig}}+u_{ig}^{(r)}\mathrm{\tilde{E}}_{1_{\bY|ig}}\right)\balpha^\top_{\bY|g}\bSigma^{-1}_{\bY|g}\balpha_{\bY|g}\right]\notag\\
        &=\sum_{i=1}^n z_{ig}^{(r)}\left(\left(1-u_{ig}^{(r)}\right)+\frac{u_{ig}^{(r)}}{\sqrt{\eta_{\bY|g}}}\right) \bSigma^{-1}_{\bY|g}\left(\by_i-\bbeta_g^\top \bx_{i}^*\right)\notag\\
        & -\frac{1}{2}\sum_{i=1}^nz_{ig}^{(r)}\left(\left(1-u_{ig}^{(r)}\right)\mathrm{E}_{1_{\bY|ig}}+u_{ig}^{(r)}\mathrm{\tilde{E}}_{1_{\bY|ig}}\right)\left(2\bSigma^{-1}_{\bY|g}\balpha_{\bY|g}\right)\label{eq alphaY partial derivative}.
    \end{align}
Equating \eqref{eq alphaY partial derivative} to the null vector yields
\begin{align}
    \balpha_{\bY|g}=\frac{\sum_{i=1}^nz_{ig}^{(r)}\left(\left(1-u_{ig}^{(r)}\right)+\frac{u_{ig}^{(r)}}{\sqrt{\eta_{\bY|g}}}\right) \left(\by_i-\bbeta_g^\top \bx_{i}^*\right)}{\sum_{i=1}^nz_{ig}^{(r)}\left(\left(1-u_{ig}^{(r)}\right)\mathrm{E}_{1_{\bY|ig}}+u_{ig}^{(r)}\mathrm{\tilde{E}}_{1_{\bY|ig}}\right)}\label{eq alphay appendix}.
\end{align}
An update for $\bbeta$ is obtained by taking the first partial derivative
 \begin{align}
        &\frac{\partial Q_{5}\left(\bbeta,\bSigma_{\bY},\balpha_{\bY},\boldeta_{\bY}|\btheta^{(r)}\right)}{\partial\bbeta_g}\notag\\
        &=\frac{\partial}{\partial\bbeta_g}\left[-\frac{1}{2}\sum_{i=1}^n\sum_{g=1}^Gz_{ig}^{(r)}\left(\left(1-u_{ig}^{(r)}\right)\mathrm{E}_{2_{\bY|ig}}+\frac{u_{ig}^{(r)}}{\eta_{\bY|g}}\mathrm{\tilde{E}}_{2_{\bY|ig}}\right)\left(\by_i-\bbeta_g^\top \bx_{i}^*\right)^\top\bSigma^{-1}_{\bY|g}\left(\by_i-\bbeta_g^T\bx_{i}^*\right)\right.\notag\\
       &+\left.\sum_{i=1}^n\sum_{g=1}^Gz_{ig}^{(r)}\left(\left(1-u_{ig}^{(r)}\right)+\frac{u_{ig}^{(r)}}{\sqrt{\eta_{\bY|g}}}\right) \left(\by_i-\bbeta_g^\top \bx_{i}^*\right)^\top\bSigma^{-1}_{\bY|g}\balpha_{\bY|g}\right]\notag\\
        &=-\frac{1}{2}\sum_{i=1}^nz_{ig}^{(r)}\left(\left(1-u_{ig}^{(r)}\right)\mathrm{E}_{2_{\bY|ig}}+\frac{u_{ig}^{(r)}}{\eta_{\bY|g}}\mathrm{\tilde{E}}_{2_{\bY|ig}}\right)\left(-2\bSigma_{\bY|g}^{-1}\left(\by_i-\bbeta_g^\top \bx_{i}^*\right)\bx_i{^*}{^\top}\right)\notag\\
       &+\sum_{i=1}^nz_{ig}^{(r)}\left(\left(1-u_{ig}^{(r)}\right)+\frac{u_{ig}^{(r)}}{\sqrt{\eta_{\bY|g}}}\right) \left(-\bSigma^{-1}_{\bY|g}\balpha_{\bY|g}\bx_i{^*}^\top\right)\notag\\
        &=\sum_{i=1}^nz_{ig}^{(r)}\left(\left(1-u_{ig}^{(r)}\right)\mathrm{E}_{2_{\bY|ig}}+\frac{u_{ig}^{(r)}}{\eta_{\bY|g}}\mathrm{\tilde{E}}_{2_{\bY|ig}}\right)\left(\by_i-\bbeta_g^\top \bx_{i}^*\right)\bx_i{^*}{^\top}\notag\\
       &-\sum_{i=1}^nz_{ig}^{(r)}\left(\left(1-u_{ig}^{(r)}\right)+\frac{u_{ig}^{(r)}}{\sqrt{\eta_{\bY|g}}}\right) \balpha_{\bY|g}\bx_i{^*}^\top\label{eq beta first derivative}
    \end{align}
and substituting $\balpha_{\bY|g}$ in \eqref{eq alphay appendix} back into $\frac{\partial Q_{5}\left(\bbeta,\bSigma_{\bY},\balpha_{\bY},\boldeta_{\bY}|\btheta^{(r)}\right)}{\partial\bbeta_g}$ in \eqref{eq beta first derivative}, which yields

\begin{align}
    &\sum_{i=1}^nz_{ig}^{(r)}\left(\left(1-u_{ig}^{(r)}\right)\mathrm{E}_{2_{\bY|ig}}+\frac{u_{ig}^{(r)}}{\eta_{\bY|g}}\mathrm{\tilde{E}}_{2_{\bY|ig}}\right)\by_i\bx_i{^*}{^\top}-\sum_{i=1}^nz_{ig}^{(r)}\left(\left(1-u_{ig}^{(r)}\right)\mathrm{E}_{2_{\bY|ig}}+\frac{u_{ig}^{(r)}}{\eta_{\bY|g}}\mathrm{\tilde{E}}_{2_{\bY|ig}}\right)\bbeta_g^\top\bx_i\bx_i{^*}{^\top}\notag\\
    &=\sum_{i=1}^nz_{ig}^{(r)}\left(\left(1-u_{ig}^{(r)}\right)+\frac{u_{ig}^{(r)}}{\sqrt{\eta_{\bY|g}}}\right) \left(\frac{\sum_{k=1}^nz_{kj}^{(r)}\left(\left(1-u_{kj}^{(r)}\right)+\frac{u_{kj}^{(r)}}{\sqrt{\eta_{\bY|g}}}\right) \left(\by_k-\bbeta_k^\top \bx_{k}^*\right)}{\sum_{k=1}^nz_{kj}^{(r)}\left(\left(1-u_{kj}^{(r)}\right)\mathrm{E}_{1_{\bY|kj}}+u_{kj}^{(r)}\mathrm{\tilde{E}}_{1_{\bY|kj}}\right)}\right)\bx_i{^*}^\top\label{eq beta first derivative alphaY sub in}.
\end{align}
Using the notation defined in \eqref{eq notation D}--\eqref{eq notation G}, \eqref{eq beta first derivative alphaY sub in} can be rewritten as
\begin{align}
    \sum\limits^n_{k=1}D_i\by_i{\bx_i^*}^\top-\frac{\sum\limits^n_{i=1}G_i\left(\sum\limits^n_{k=1}G_k\by_k\right){\bx_i^*}^\top}{\sum\limits^n_{k=1}F_k}&=
    \sum\limits^n_{k=1}D_i\bbeta_g^\top\bx_i^*{\bx_i^*}^\top-\frac{\sum\limits^n_{i=1}G_i\left(\sum\limits^n_{k=1}G_k\bbeta_g^\top\bx_k^*\right){\bx_i^*}^\top}{\sum\limits^n_{k=1}F_k}\notag\\
     \sum\limits^n_{k=1}D_i\by_i{\bx_i^*}^\top-\frac{\sum\limits^n_{i=1}G_i\left(\sum\limits^n_{k=1}G_k\by_k\right){\bx_i^*}^\top}{\sum\limits^n_{k=1}F_k}&=\bbeta_g^\top\left(
    \sum\limits^n_{k=1}D_i\bx_i^*{\bx_i^*}^\top-\frac{\sum\limits^n_{i=1}G_i\left(\sum\limits^n_{k=1}G_k\bx_k^*\right){\bx_i^*}^\top}{\sum\limits^n_{k=1}F_k}\right)\notag\\
    \left( \sum\limits^n_{k=1}D_i\by_i{\bx_i^*}^\top-\frac{\sum\limits^n_{i=1}G_i\left(\sum\limits^n_{k=1}G_k\by_k\right){\bx_i^*}^\top}{\sum\limits^n_{k=1}F_k}\right)^\top&=\left(
    \sum\limits^n_{k=1}D_i\bx_i^*{\bx_i^*}^\top-\frac{\sum\limits^n_{i=1}G_i\left(\sum\limits^n_{k=1}G_k\bx_k^*\right){\bx_i^*}^\top}{\sum\limits^n_{k=1}F_k}\right)^\top \bbeta_g\label{eq beta first derivative final}. 
\end{align}
By equating \eqref{eq beta first derivative final} to the null vector, and solving for $\bbeta$, we obtain the update.
\end{proof}

\end{document}